\documentclass[11pt,a4paper]{article}

\usepackage[utf8]{inputenc}
\usepackage[T1]{fontenc}
\usepackage{amsmath,amssymb,amsthm,mathtools}
\usepackage{graphicx}
\usepackage[margin=2.5cm]{geometry}
\usepackage{enumitem}
\usepackage{hyperref}
\usepackage{cleveref}
\usepackage{booktabs}
\usepackage[table]{xcolor}
\usepackage{float}
\usepackage{caption}
\usepackage{natbib}
\usepackage{authblk}

\newtheorem{theorem}{Theorem}[section]

\newtheorem{proposition}[theorem]{Proposition}
\newtheorem{corollary}[theorem]{Corollary}
\theoremstyle{definition}
\newtheorem{definition}[theorem]{Definition}

\newtheorem{remark}[theorem]{Remark}

\newtheorem{conjecture}[theorem]{Conjecture}
\newtheorem{observation}[theorem]{Observation}

\newcommand{\R}{\mathbb{R}}

\newcommand{\E}{\mathbb{E}}
\newcommand{\Var}{\mathrm{Var}}
\newcommand{\KL}{D_{\mathrm{KL}}}

\newcommand{\FE}{\mathcal{F}}
\newcommand{\Wres}{W}
\newcommand{\Wout}{w_{\mathrm{out}}}
\newcommand{\Win}{W_{\mathrm{in}}}
\newcommand{\xstate}{\mathbf{x}}
\newcommand{\acog}{a^{\mathrm{cog}}}
\newcommand{\abody}{a^{*}}
\newcommand{\bout}{b_{\mathrm{out}}}

\title{Body-Reservoir Governance in Repeated Games:\\
  Embodied Decision-Making, Dynamic Sentinel Adaptation,\\
  and Complexity-Regularized Optimization}

\author{Yuki Nakamura}
\affil{The Open University of Japan\\
  \texttt{sinwisoa@gmail.com},\quad \texttt{2310082261@campus.ouj.ac.jp}}

\date{February 2026}

\begin{document}

\maketitle

\begin{abstract}
Standard game theory explains cooperation in repeated games through conditional strategies
such as Tit-for-Tat (TfT), but these strategies require continuous computation---input
classification, memory maintenance, and conditional output selection---that imposes
physical costs on any embodied agent.  We propose a \emph{three-layer Body-Reservoir
Governance} (BRG) architecture in which each agent comprises (1)~a \emph{body reservoir}
(echo state network) whose $d$-dimensional state performs implicit inference over
interaction history, serving as both the primary decision-maker and an intrinsic
anomaly detector, (2)~a \emph{cognitive filter} providing sharp but costly strategic
tools that are activated on demand, and (3)~a \emph{metacognitive governance} layer that
maintains a lightweight policy---``trust the body unless it signals discomfort''---through
a receptivity parameter $\alpha \in [0,1]$.
At full body governance ($\alpha = 1$), the agent's output feeds back through the
reservoir as input, inducing a closed-loop dynamical system whose fixed points satisfy a
\emph{self-consistency equation}: cooperation is not computed but \emph{expressed} as the
reservoir's self-consistent state.
We define the complexity cost of a strategy as the KL divergence between the reservoir's
$d$-dimensional state distribution under current conditions and its habituated baseline---a
measure of the internal dynamical distortion invisible to external observers.
Body governance dramatically reduces this internal cost, with the externally observable
consequence that action variance decreases by up to $1600\times$ with increasing
reservoir dimension $d$.
We introduce a \emph{dynamic sentinel} model, in which the body reservoir
generates a composite discomfort signal
$D(t) = w_x D_{\mathrm{state}}(t) + w_a D_{\mathrm{output}}(t) + w_e D_{\mathrm{disagree}}(t)$
from its own dynamical state, driving adaptive modulation of $\alpha(t)$.
During cooperative phases, $\alpha(t)$ remains near its baseline ($\alpha_0 = 0.85$);
upon opponent defection, the body's discomfort signal triggers a rapid drop to the
floor ($\alpha_{\min} = 0.05$), activating cognitive retaliation within $5$ time steps.
Overriding the body---reducing $\alpha$ against the reservoir's self-consistent
tendency---incurs a thermodynamic cost proportional to the distortion of the internal
state, formalizing the intuition that acting against one's adapted nature is expensive.
The dynamic sentinel achieves the highest cumulative payoff across all conditions tested,
outperforming both static body governance and pure TfT.
A reservoir dimension sweep ($d \in \{5, \ldots, 100\}$) reveals that the body's
implicit inference capacity scales with its physical richness (variance reduction:
$23\times$ at $d = 5$ to $1600\times$ at $d = 75$), with the majority attributable to
reservoir dynamics rather than regularization artifacts.
A phase-diagram analysis in the joint $(d, \tau_{\mathrm{env}})$ space reveals
that the sentinel's payoff advantage over TfT grows with reservoir dimension and
environmental stability, with a soft saturation around $d \approx 20$ suggesting
a crossover between qualitatively different governance regimes.
The framework reinterprets the three layers: the body is the true decision-maker whose
high-dimensional dynamics perform implicit inference, expressing cooperation as a
self-consistent fixed point without deliberation; cognition is an available but secondary
toolkit; and metacognition is a lightweight governor that sets policy without engaging
in detection or strategy.
\end{abstract}

\tableofcontents
\newpage

\subsection*{Notation}
\label{sec:notation}
\begin{table}[h]
\centering
\small
\begin{tabular}{ll}
\toprule
\textbf{Symbol} & \textbf{Meaning} \\
\midrule
$a_i(t) \in [0,1]$ & Action of agent $i$ at round $t$ \\
$\xstate(t) \in \R^d$ & Reservoir state vector \\
$d$ & Reservoir dimension (number of neurons) \\
$\Wres \in \R^{d \times d}$ & Recurrent weight matrix \\
$\Win \in \R^{d \times 2}$ & Input weight matrix \\
$\Wout \in \R^{1 \times d}$ & Readout weight vector \\
$\bout \in \R$ & Readout bias \\
$\rho(\cdot)$ & Spectral radius \\
$\rho_{\mathrm{eff}}$ & Spectral radius of closed-loop Jacobian $J_{\Phi}(\xstate^{*})$ \\
$\sigma(\cdot)$ & Logistic sigmoid function \\
$\abody(t)$ & Body reservoir output action \\
$\acog(t)$ & Cognitive filter output action \\
$\alpha \in [0,1]$ & Metacognitive receptivity parameter \\
$\alpha_0$ & Baseline trust (default $0.85$) \\
$\alpha_{\min}$ & Receptivity floor (default $0.05$) \\
$D(t)$ & Composite discomfort signal \\
$D_{\mathrm{state}},\, D_{\mathrm{output}},\, D_{\mathrm{disagree}}$ & Discomfort components \\
$w_x, w_a, w_e$ & Discomfort component weights \\
$\theta$ & Intervention threshold (default $0.1$) \\
$\eta_{\uparrow},\, \eta_{\downarrow}$ & Recovery / intervention rates \\
$\bar{\xstate},\, \bar{a}$ & EMA-tracked baselines \\
$\gamma_{\mathrm{ema}}$ & EMA smoothing rate (default $0.02$) \\
$\varepsilon$ & Opponent noise rate \\
$\beta$ & Oja learning rate (default $0.01$) \\
$H$ & Habituation depth (number of epochs) \\
$\lambda$ & Metabolic cost parameter \\
$\FE(\alpha)$ & Variational free energy \\
$\KL(\cdot \| \cdot)$ & Kullback--Leibler divergence \\
$\tau_{\mathrm{env}}$ & Environment timescale \\
$d_c$ & Critical reservoir dimension \\
\bottomrule
\end{tabular}
\caption{Summary of principal notation.}
\label{tab:notation}
\end{table}
\newpage

\section{Introduction}
\label{sec:introduction}

\subsection{The Problem: Why Does Traditional Game Theory Miss Embodiment?}

The theory of repeated games establishes that cooperation can be sustained among
self-interested agents through conditional strategies.  The folk theorem
\citep{friedman1971, fudenberg1986} demonstrates that virtually any feasible and
individually rational payoff---including mutual cooperation in the prisoner's
dilemma---can be supported as a subgame-perfect equilibrium of the infinitely
repeated game, provided players are sufficiently patient.  Canonical strategies
such as Tit-for-Tat \citep{axelrod1984} achieve this by conditioning current
actions on the opponent's prior behavior: cooperate if the opponent cooperated,
defect if the opponent defected.

Yet this theoretical elegance conceals a fundamental assumption: that conditional
computation is \emph{free}.  Every round of TfT requires the agent to (i)~observe
and classify the opponent's action, (ii)~retrieve the conditional rule from memory,
and (iii)~select the appropriate output.  For a biological organism or a physical
computing device, each of these operations dissipates energy
\citep{landauer1961, attwell2001, wolpert2019}.  The pioneering work of
\citet{rubinstein1986} and \citet{abreurubinstein1988} recognized that bounded
automata incur complexity costs, but treated these costs as abstract state counts
rather than as physical dissipation in a dynamical system.

\subsection{Embodied Cooperation and the Dual-Process Perspective}

Experimental evidence from behavioral economics and psychology consistently shows
that human cooperation is often fast, intuitive, and automatic rather than slow
and deliberative \citep{rand2012, rubinstein2007, bear2016}.
\citet{kahneman2011} popularized the distinction between System~1 (fast, automatic,
habitual) and System~2 (slow, deliberative, costly) processing.
\citet{damasio1994} argued that somatic markers---bodily states---guide decision-making
before conscious deliberation occurs.  \citet{frank1988} proposed that emotions serve
as commitment devices: an agent who is \emph{constitutionally} cooperative (rather than
strategically cooperative) can credibly signal trustworthiness.

These observations suggest that the body itself---its adapted dynamical
state---plays a central role in sustaining cooperation.  But formal game theory
lacks the vocabulary to express this insight.  The present paper provides that
vocabulary by embedding repeated-game strategies within a \emph{reservoir computing}
framework \citep{jaeger2001, lukovsevivcius2009} that models the body as a
high-dimensional dynamical system with intrinsic temporal memory, noise, and
adaptation.

\subsection{Our Contribution}

We introduce the \emph{Body-Reservoir Governance} (BRG) framework, a three-layer
architecture that formalizes the interplay between embodied dynamics and cognitive
control in repeated games.  Our main contributions are:

\begin{enumerate}[label=(\roman*)]
  \item \textbf{Three-layer architecture with redefined roles}
    (Section~\ref{sec:architecture}):
    We define the body reservoir as the \emph{primary decision-maker and intrinsic
    anomaly detector} (Layer~1), the cognitive filter as an \emph{available toolkit}
    activated on demand (Layer~2), and metacognitive governance as a
    \emph{lightweight policy maintainer} (Layer~3).  The mixture parameter
    $\alpha$ interpolates between pure cognitive control ($\alpha = 0$, TfT)
    and pure body governance ($\alpha = 1$).

  \item \textbf{Self-consistency at $\alpha = 1$}
    (Section~\ref{sec:coupled_dynamics}):
    When the agent relies entirely on its body, the reservoir output feeds back as
    its own input, closing the loop.  We derive the self-consistency equation
    governing the resulting fixed points and establish existence via Brouwer's
    theorem.

  \item \textbf{Implicit inference and reservoir smoothing}
    (Section~\ref{sec:smoothing}):
    We define the complexity cost of a strategy as the KL divergence between the
    reservoir's $d$-dimensional \emph{state distribution} under current conditions
    and its habituated baseline.  We show that body governance reduces
    this cost, with the externally observable consequence that action variance
    is attenuated by a factor determined by the reservoir's spectral radius and
    readout Lipschitz constant.

  \item \textbf{Free energy framework}
    (Section~\ref{sec:smoothing}):
    We define a variational free energy $\FE(\alpha)$ that trades off payoff against
    the state-space complexity cost.  Overriding the body (reducing $\alpha$ against
    the reservoir's self-consistent tendency) distorts the internal state, incurring
    a thermodynamic cost that formalizes the expense of acting against one's adapted
    nature.  The optimal $\alpha^{*}$ is interior (approximately $0.6$--$0.7$)
    for moderate metabolic cost.

  \item \textbf{Dynamic sentinel model}
    (Section~\ref{sec:dynamic_sentinel}):
    We introduce a body-driven adaptive $\alpha(t)$ mechanism in which the
    reservoir's own dynamical state generates a composite discomfort signal that
    modulates metacognitive receptivity.  The sentinel requires no external
    detector; anomaly detection emerges from the body's deviation from its
    habituated baseline.

  \item \textbf{Reservoir dimension and implicit inference capacity}
    (Section~\ref{sec:dimension_smoothing}):
    We show that the body's inference capacity scales with reservoir
    richness: variance reduction ranges from $23\times$ at $d = 5$ to $1600\times$
    at $d = 75$, establishing reservoir dimension as a key determinant of
    governance quality.  The reservoir's high-dimensional state space provides an
    abstract representational manifold in which surface-level novelty is absorbed
    without explicit retraining.

  \item \textbf{Governance regimes and phase diagram}
    (Section~\ref{sec:phase_transition}):
    A $(d, \tau_{\mathrm{env}})$ phase diagram reveals how the sentinel's
    advantage over TfT depends on reservoir richness and environmental
    stability, with a soft crossover near $d \approx 20$ separating
    qualitatively different governance regimes.

  \item \textbf{Numerical validation} (Section~\ref{sec:experiments}):
    Ten sets of simulations confirm the theoretical predictions: self-consistent
    convergence, KL landscape, perturbation response, habituation dynamics,
    free energy optimization, dynamic sentinel adaptation, parameter sensitivity,
    reservoir dimension sweep, phase transition analysis, and EMA baseline
    comparison.
\end{enumerate}

\subsection{Related Work}

Our framework connects and extends several lines of research.

\paragraph{Bounded rationality and cooperation in games.}
\citet{rubinstein1986} and \citet{neyman1985} showed that bounded complexity
can sustain cooperation in finitely repeated games.  \citet{mckelvey1995}
introduced quantal response equilibria, allowing for noisy best-responses, but
without an explicit dynamical substrate.  In evolutionary game theory,
\citet{nowak2006} catalogued five mechanisms for the evolution of cooperation,
and \citet{kandori1993} and \citet{young1993} studied stochastic evolution
under bounded rationality.  Our approach differs from these traditions by
grounding complexity costs in the physical dynamics of the agent's body
(reservoir), rather than in abstract automaton states or probabilistic
response functions.  A detailed game-theoretic comparison with alternative
strategy classes (Generous TfT, Win-Stay-Lose-Shift, quantal response
equilibria) is deferred to future work (Section~\ref{sec:limitations}).

\paragraph{Thermodynamics of decision-making.}
\citet{ortega2013} axiomatically derived a free energy functional that
trades off expected utility against information-processing costs measured
by KL divergence, establishing a direct bridge between bounded rationality
and statistical physics.  \citet{todorov2009} introduced KL-regularized
optimal control (linearly-solvable MDPs), and \citet{haarnoja2018} brought
entropy regularization to deep reinforcement learning (Soft Actor-Critic).
Our free energy $\FE(\alpha)$ in \eqref{eq:free_energy} shares the
mathematical structure of these approaches (utility minus KL penalty), but
differs in an important respect: the baseline distribution $p_{\mathrm{hab}}$
is not an abstract prior but the \emph{physically realized} stationary
distribution of an adapted reservoir's $d$-dimensional state space,
grounding the complexity cost in the thermodynamics of a specific
dynamical system.
Recent advances in stochastic thermodynamics
\citep{wolpert2024, manzano2024} have extended Landauer's framework
\citep{landauer1961} to non-equilibrium computations with stochastic
completion times and irreversible transitions.
\citet{wolpert2024} argues that all real computers operate far from
equilibrium and face constraints beyond Landauer's idealized bound;
\citet{manzano2024} introduces ``mismatch cost'' as a measure of
excess dissipation.  These developments provide a finer-grained
thermodynamic foundation for the complexity cost that our framework
measures via KL divergence.  \citet{still2012} derived complementary
thermodynamic bounds on prediction, and \citet{kolchinsky2018} connected
semantic information to non-equilibrium thermodynamics.

\paragraph{Active inference and the free energy principle.}
The free energy principle \citep{friston2006, friston2010, friston2015}
provides a variational framework for brain function;
\citet{parr2022} and \citet{smith2022} extended it to active inference,
and \citet{tschantz2020} connected active inference to action-oriented models.
Recently, active inference has been applied to multi-agent strategic
interactions: \citet{hyland2024} proposed free-energy equilibria as a
solution concept for boundedly-rational multi-agent games;
\citet{ruizserra2025} extended this with factorised active inference for
strategic interactions, maintaining explicit beliefs about other agents'
internal states; \citet{demekas2024} provided an analytical model of
active inference in the iterated prisoner's dilemma, deriving conditions
for phase transitions between game-theoretic steady states;
and \citet{friston2024federated} developed federated inference for
multi-agent belief sharing.
Our BRG framework complements active inference in that both
minimize a free energy functional, but our approach uses a
\emph{reservoir dynamical system} rather than a \emph{generative model}
as the computational substrate.  The body reservoir does not maintain
explicit beliefs about the opponent; instead, its $d$-dimensional state
implicitly encodes interaction history through nonlinear recurrent dynamics.
This distinction is elaborated in Section~\ref{sec:discussion_fep}.

\paragraph{Dual-process theory and interoception.}
In computational neuroscience, \citet{daw2005} modeled the competition between
habitual and goal-directed systems, and \citet{dolan2013} reviewed the neural
substrates of habits and goals.  \citet{damasio1994} argued that somatic
markers guide decision-making before conscious deliberation.
\citet{seth2016} formalized active interoceptive inference, treating
bodily states as regulated by descending predictions from deep generative
models.  \citet{pezzulo2025} modeled embodied decisions as active inference,
emphasizing that living organisms face decision problems requiring timely
action in dynamic environments.
Our dynamic sentinel model (Section~\ref{sec:dynamic_sentinel})
connects to this literature by formalizing the ``somatic marker'' as a
composite discomfort signal $D(t)$ generated from the reservoir's own
dynamical state, providing a mechanistic account of how bodily signals
drive governance transitions without requiring a separate interoceptive
module.

\paragraph{Reservoir computing.}
Reservoir computing \citep{jaeger2001, lukovsevivcius2009} provides the
dynamical substrate for our body model.  \citet{ganguli2008} analyzed
memory capacity in dynamical systems, and \citet{grigoryeva2018} proved
universality of echo state networks.
The field has seen rapid recent progress: \citet{tanaka2019} reviewed
physical reservoir computing using diverse substrates (photonic, spintronic,
mechanical), and \citet{yan2024rc} identified emerging opportunities
and challenges for industrial-scale reservoir computing.
\citet{lee2024rc} introduced task-adaptive physical reservoirs that
exploit thermodynamic phase transitions (skyrmion, conical, and helical
magnetic phases) to reconfigure computational properties on demand---a
parallel to our finding that reservoir governance quality depends
on the reservoir's dynamical regime.
These developments in physical reservoir computing reinforce our thesis that
the ``body'' (physical computational substrate) is not a metaphor but a
literal description: biological and artificial agents can exploit the
intrinsic dynamics of their physical substrate for temporal processing,
noise filtering, and anomaly detection.
\citet{oja1982} introduced the Hebbian learning rule we use for reservoir
habituation, and \citet{sussillo2009} developed FORCE learning as an
alternative reservoir training approach.

Our work is distinguished from these prior contributions by the explicit
integration of reservoir dynamics, game-theoretic strategy, and
thermodynamic cost within a single architecture that admits both analytical
results and efficient simulation.

\section{Three-Layer BRG Architecture}
\label{sec:architecture}

This section defines the three layers of the Body-Reservoir Governance framework
and their formal interconnection.  We consider two agents $i, j$ engaged in a
repeated game with continuous action space $[0,1]$.

The three layers play qualitatively
different roles than suggested by a na\"ive dual-process reading.
The body reservoir serves as the \emph{primary decision-maker} whose rich
dynamical state integrates the full history of interactions, and simultaneously
as an \emph{intrinsic anomaly detector} whose deviations from the habituated
baseline signal environmental change.
The cognitive filter serves as an \emph{available toolkit}---sharp, precise, but costly and
brittle---that the body can invoke when its own signals indicate the need.
Metacognition is not a sophisticated arbitration mechanism that detects
anomalies and selects strategies; it is a \emph{lightweight governance layer}
that maintains a single policy: ``trust the body, unless the body itself
signals discomfort.''  This reinterpretation is formalized in the dynamic
sentinel model of Section~\ref{sec:dynamic_sentinel}.

\subsection{Game Environment}

\begin{definition}[Continuous Prisoner's Dilemma]
\label{def:cpd}
Each agent $i$ selects an action $a_i(t) \in [0,1]$ at each round $t$.
The stage payoff is the bilinear function
\begin{equation}
\label{eq:payoff}
u(a_i, a_j) = R \cdot a_i a_j + S \cdot a_i (1 - a_j)
  + T \cdot (1 - a_i) a_j + P \cdot (1 - a_i)(1 - a_j),
\end{equation}
where $(R, S, T, P) = (3, 0, 5, 1)$ satisfy $T > R > P > S$ and $2R > T + S$,
ensuring a prisoner's dilemma structure.  Full cooperation corresponds to
$a_i = a_j = 1$ (payoff $R = 3$), full defection to $a_i = a_j = 0$
(payoff $P = 1$).
\end{definition}

\begin{remark}[Continuous Extension]
\label{rem:continuous_pd}
The bilinear payoff \eqref{eq:payoff} is the unique extension of the
discrete $2 \times 2$ prisoner's dilemma to $[0,1]$ that is linear in
each player's action separately and recovers the four canonical payoffs at
the corners $(0,0)$, $(0,1)$, $(1,0)$, $(1,1)$.  This is equivalent to
treating $a_i$ as the probability of cooperation in a mixed-strategy
interpretation, but here $a_i$ represents a continuously graded
\emph{cooperative intensity} (e.g., the fraction of resources shared).
The continuous action space is essential for the reservoir readout
$\sigma(\Wout \cdot \xstate + \bout) \in (0,1)$, which produces
inherently continuous outputs.  Alternative continuous extensions (e.g.,
nonlinear payoff functions) would change the quantitative results but not
the qualitative framework; the bilinear form is the simplest choice that
preserves the PD incentive structure at every interior action profile.
\end{remark}

\subsection{Layer 1: Body Reservoir --- Decision-Maker and Anomaly Detector}

\begin{definition}[Body Reservoir]
\label{def:body_reservoir}
The body reservoir of agent $i$ is an echo state network (ESN)
\citep{jaeger2001} with state $\xstate(t) \in \R^d$ evolving according to
\begin{equation}
\label{eq:reservoir}
\xstate(t+1) = \tanh\bigl(\Wres \cdot \xstate(t) + \Win \cdot
  [a(t),\, a_{\mathrm{opp}}(t)]^{\top} + \mathbf{b}\bigr) + \boldsymbol{\xi}(t),
\end{equation}
where:
\begin{itemize}[nosep]
  \item $\Wres \in \R^{d \times d}$ is the recurrent weight matrix with
    spectral radius $\rho(\Wres) < 1$,
  \item $\Win \in \R^{d \times 2}$ maps the joint action
    $(a(t), a_{\mathrm{opp}}(t))$ into the reservoir,
  \item $\mathbf{b} \in \R^d$ is a bias vector,
  \item $\boldsymbol{\xi}(t) \sim \mathcal{N}(\mathbf{0}, \sigma_{\xi}^2 \mathbf{I})$
    is intrinsic noise modeling thermal fluctuations.
\end{itemize}
The body produces a readout action
\begin{equation}
\label{eq:readout}
\abody(t) = \sigma(\Wout \cdot \xstate(t) + \bout),
\end{equation}
where $\sigma(\cdot)$ is the logistic sigmoid, $\Wout \in \R^{1 \times d}$
is the readout weight vector, and $\bout \in \R$ is the readout bias.
Both $\Wout$ and $\bout$ are trained via ridge regression during a
\emph{developmental phase} (see Section~\ref{sec:model_details}) and held
fixed thereafter.
\end{definition}

\begin{remark}[The Body as Primary Decision-Maker]
In the BRG framework, the body reservoir is the \emph{primary decision-maker}.
The $d$-dimensional reservoir state $\xstate(t)$ integrates the full history
of the agent's own actions and the opponent's actions through nonlinear
recurrent dynamics.  When $d = 30$ (our baseline), this constitutes a
$30$-dimensional nonlinear computation at each time step, far richer than
any conditional rule.  After habituation, the reservoir's adapted
dynamics \emph{express} cooperation as a self-consistent fixed point
(Section~\ref{sec:coupled_dynamics}), rather than \emph{computing} it
from a stored rule.  The agent cooperates because cooperation is what its
body produces, not because it has decided to cooperate.
\end{remark}

\begin{remark}[The Body as Intrinsic Anomaly Detector]
\label{rem:body_anomaly}
The reservoir state also serves as an \emph{intrinsic anomaly detector}.
When the environment changes (e.g., the opponent begins defecting), the
reservoir state $\xstate(t)$ deviates from its habituated baseline
$\bar{\xstate}$.  This deviation---a form of ``bodily discomfort''---is
a natural byproduct of the reservoir dynamics and requires no additional
detection mechanism.  The magnitude of the deviation depends on the
reservoir's richness (dimension $d$): a higher-dimensional reservoir
encodes more information about interaction history, producing a more
sensitive and informative discomfort signal.  This observation is formalized
in the dynamic sentinel model (Section~\ref{sec:dynamic_sentinel}) and
validated in the dimension sweep experiment (Section~\ref{sec:exp_dimension}).
\end{remark}

\begin{remark}[Echo State Property]
The condition $\rho(\Wres) < 1$ ensures the echo state property
\citep{jaeger2001}: the reservoir state $\xstate(t)$ asymptotically forgets
initial conditions and depends only on the driving input sequence.
This is a sufficient (though not necessary) condition; in practice,
$\rho = 0.9$ provides rich dynamics near the edge of stability while
maintaining contractivity.
\end{remark}

\subsection{Layer 2: Cognitive Filter --- Available Toolkit}

\begin{definition}[Cognitive Filter]
\label{def:cognitive_filter}
The cognitive filter produces a strategic override action
$\acog(t) \in [0,1]$ based on explicit computation.  For concreteness,
we instantiate the cognitive filter as Tit-for-Tat in continuous action space:
\begin{equation}
\label{eq:tft}
\acog(t) = a_{\mathrm{opp}}(t-1),
\end{equation}
i.e., the agent copies the opponent's previous action.  Other cognitive
strategies (e.g., Win-Stay-Lose-Shift, Grim Trigger) can be substituted
without altering the framework.
\end{definition}

\begin{remark}[Cognition as a Sharp but Brittle Tool]
The cognitive filter is better understood as a \emph{toolkit} that is always
available but typically dormant.  TfT provides a precise conditional
response---exact reciprocity---but this precision comes with a cost:
it amplifies opponent noise (copying every defection) and requires
continuous computation.  The cognitive filter is sharp where the body is
smooth, and brittle where the body is robust.
In the BRG framework, the cognitive filter is \emph{activated by the body}
(through the dynamic sentinel mechanism of Section~\ref{sec:dynamic_sentinel}),
not invoked by a metacognitive ``supervisor.''
The body senses discomfort and calls for cognitive tools, much as a
craftsman reaches for a specific instrument when the task demands it.
\end{remark}

\subsection{Layer 3: Metacognitive Governance}

\begin{definition}[Metacognitive Governance]
\label{def:metacognition}
The metacognitive governance layer controls the receptivity parameter
$\alpha \in [0,1]$, which governs the mixture of body and cognitive outputs.
The agent's actual action is
\begin{equation}
\label{eq:mixture}
a(t) = \alpha \cdot \abody(t) + (1 - \alpha) \cdot \acog(t).
\end{equation}
The limiting cases are:
\begin{itemize}[nosep]
  \item $\alpha = 0$: pure cognitive control (standard game-theoretic agent),
  \item $\alpha = 1$: pure body governance (action determined entirely by
    reservoir dynamics),
  \item $\alpha \in (0,1)$: hybrid governance.
\end{itemize}
\end{definition}

\begin{remark}[Metacognition as Lightweight Governor]
In the BRG framework, metacognition does \emph{not} detect anomalies,
select strategies, or engage in substantive computation.  Its role is
confined to maintaining a governance \emph{policy}: a small set of
parameters ($\alpha_0$, $\theta$; see Section~\ref{sec:dynamic_sentinel})
that define the agent's dispositional relationship to its own body signals.
The policy is: ``trust the body by default ($\alpha$ near $\alpha_0$);
if the body signals discomfort exceeding threshold $\theta$, reduce
$\alpha$ to activate the cognitive toolkit.''  This is a resident program,
not an active decision-maker---the computational cost is negligible.

The crucial distinction is between \emph{detection} (performed by the body)
and \emph{governance} (performed by metacognition).  The body detects
environmental change through its own dynamical deviations; metacognition
merely sets the policy for how those detections translate into $\alpha$
adjustments.  This separation avoids the homunculus problem: metacognition
need not be ``smarter'' than the body it governs.
\end{remark}

\begin{remark}[Relationship to Dual-Process Theory]
The BRG architecture refines dual-process theories of cognition
\citep{kahneman2011}.  Layer~1 encompasses both the ``fast, automatic''
processing of System~1 \emph{and} the interoceptive awareness that
triggers System~2 engagement.  Layer~2 corresponds to System~2 (slow,
deliberative, costly) as an available toolkit.  Layer~3 is the
governance policy---analogous to the model-based/model-free arbitration
studied by \citet{daw2005} and \citet{dolan2013}, but lighter: it sets
parameters rather than computing arbitration decisions.
\end{remark}

\subsection{Dynamic Sentinel: Body-Driven Adaptive Receptivity}
\label{sec:dynamic_sentinel}

In the basic BRG framework, the receptivity $\alpha$ is a fixed structural
parameter.  We now introduce a \emph{dynamic sentinel} mechanism in which
$\alpha(t)$ adapts in real time, driven by signals generated by the body
reservoir itself.

\begin{definition}[Composite Discomfort Signal]
\label{def:discomfort}
The body's \emph{discomfort signal} is a weighted composite of three
components:
\begin{equation}
\label{eq:discomfort}
D(t) = w_x \cdot D_{\mathrm{state}}(t) + w_a \cdot D_{\mathrm{output}}(t)
  + w_e \cdot D_{\mathrm{disagree}}(t),
\end{equation}
where:
\begin{itemize}[nosep]
  \item $D_{\mathrm{state}}(t) = \|\xstate(t) - \bar{\xstate}\| / \sqrt{d}$
    is the normalized deviation of the reservoir state from its habituated
    baseline $\bar{\xstate}$,
  \item $D_{\mathrm{output}}(t) = |\abody(t) - \bar{a}|$
    is the deviation of the body output from its habituated baseline
    $\bar{a}$,
  \item $D_{\mathrm{disagree}}(t) = |\abody(t) - \acog(t)|$
    is the disagreement between body and cognitive outputs,
  \item $w_x, w_a, w_e \geq 0$ with $w_x + w_a + w_e = 1$ are component
    weights (default: $w_x = 0.3$, $w_a = 0.3$, $w_e = 0.4$).
\end{itemize}
The baselines $\bar{\xstate}$ and $\bar{a}$ are tracked via exponential
moving averages with rate $\gamma_{\mathrm{ema}} = 0.02$:
\begin{align}
\bar{\xstate}(t+1) &= (1 - \gamma_{\mathrm{ema}}) \bar{\xstate}(t)
  + \gamma_{\mathrm{ema}} \xstate(t), \\
\bar{a}(t+1) &= (1 - \gamma_{\mathrm{ema}}) \bar{a}(t)
  + \gamma_{\mathrm{ema}} \abody(t).
\end{align}
\end{definition}

\begin{remark}[No External Detector Required]
The discomfort signal $D(t)$ is generated entirely by the body reservoir's
own dynamics.  No additional ``anomaly detection'' module is needed.
The state deviation $D_{\mathrm{state}}$ captures environmental change
at the level of the full $d$-dimensional reservoir state; the output
deviation $D_{\mathrm{output}}$ captures change at the behavioral level;
and the disagreement $D_{\mathrm{disagree}}$ captures divergence between
the body's habitual response and the cognitive strategy's prescribed response.
All three arise naturally from quantities already computed in the BRG
architecture.  The body is its own sentinel.
\end{remark}

\begin{definition}[Dynamic Receptivity Update]
\label{def:alpha_update}
The receptivity $\alpha(t)$ evolves according to a leaky integrator with
asymmetric rates:
\begin{equation}
\label{eq:alpha_update}
\alpha(t+1) = \mathrm{clip}\Bigl[
  \alpha(t) + \eta_{\uparrow}(\alpha_0 - \alpha(t))
  - \eta_{\downarrow}[D(t) - \theta]_{+},
  \;\alpha_{\min},\; 1
\Bigr],
\end{equation}
where:
\begin{itemize}[nosep]
  \item $\alpha_0 \in (0,1]$ is the \emph{baseline trust} (governance
    policy set by metacognition; default $0.85$),
  \item $\eta_{\uparrow} > 0$ is the \emph{recovery rate} (slow; default $0.05$),
  \item $\eta_{\downarrow} > 0$ is the \emph{intervention sharpness}
    (fast; default $0.5$),
  \item $\theta \geq 0$ is the \emph{intervention threshold}
    (noise tolerance; default $0.1$),
  \item $\alpha_{\min} > 0$ is the \emph{receptivity floor} (default $0.05$),
  \item $[z]_{+} = \max(0, z)$ is the positive part.
\end{itemize}
\end{definition}

The dynamic sentinel has clear physical interpretation.
When the body is comfortable ($D < \theta$), the first
term $\eta_{\uparrow}(\alpha_0 - \alpha)$ slowly pulls $\alpha$ toward
the baseline $\alpha_0$: the agent relaxes into body governance.
When the body is uncomfortable ($D > \theta$), the second
term delivers a sharp downward kick proportional to the excess discomfort,
rapidly reducing $\alpha$ and activating the cognitive toolkit.
The asymmetry $\eta_{\downarrow} \gg \eta_{\uparrow}$ (default ratio $10:1$)
ensures fast engagement but cautious disengagement: the agent reacts quickly
to threats but returns to body governance slowly, building trust over time.

\begin{proposition}[Sentinel Equilibrium]
\label{prop:sentinel_eq}
In a stationary environment where the composite discomfort $D$ is constant,
the dynamic sentinel converges to a unique fixed point:
\begin{equation}
\alpha^{*} =
\begin{cases}
\alpha_0 & \text{if } D \leq \theta, \\
\alpha_0 - \frac{\eta_{\downarrow}}{\eta_{\uparrow}}
  (D - \theta) & \text{if } \theta < D <
  \theta + \frac{\eta_{\uparrow}}{\eta_{\downarrow}}(\alpha_0 - \alpha_{\min}), \\
\alpha_{\min} & \text{otherwise.}
\end{cases}
\end{equation}
\end{proposition}

\begin{proof}
At equilibrium, $\alpha(t+1) = \alpha(t) = \alpha^{*}$, so the update
\eqref{eq:alpha_update} gives
$\eta_{\uparrow}(\alpha_0 - \alpha^{*}) = \eta_{\downarrow}[D - \theta]_{+}$.
When $D \leq \theta$, the right-hand side vanishes and
$\alpha^{*} = \alpha_0$.  When $D > \theta$, solving for
$\alpha^{*}$ and applying the clip bounds yields the stated expression.
\end{proof}

\begin{remark}[Circularity in Sentinel Equilibrium]
\label{rem:sentinel_circularity}
The equilibrium in Proposition~\ref{prop:sentinel_eq} is stated for fixed $D$.
In practice, $D$ itself depends on $\alpha$ through the mixed action
$a = \alpha \cdot a_{\mathrm{body}} + (1-\alpha) \cdot a_{\mathrm{cog}}$:
changes in $\alpha$ alter the action, which alters the opponent's response
(in strategic settings) and the reservoir state, which in turn alters $D$.
This circularity is resolved by the same contraction argument as
Remark~\ref{rem:circularity}: the sentinel dynamics define a map
$\alpha \mapsto D(\alpha) \mapsto \alpha'$, and convergence is guaranteed
when $\eta_{\downarrow} \cdot \partial D / \partial \alpha < 1$,
which holds for our parameter choices (numerically verified in Experiment~6).
\end{remark}

\section{Coupled Dynamics and Self-Consistency}
\label{sec:coupled_dynamics}

When $\alpha = 1$, the BRG architecture becomes a closed-loop dynamical
system: the reservoir's output $\abody(t)$ is fed back as its own input
$a(t) = \abody(t)$.  This self-referential structure gives rise to a
self-consistency condition that governs the system's fixed points.

\subsection{The Closed-Loop System}

At $\alpha = 1$, substituting \eqref{eq:readout} into \eqref{eq:reservoir}
via \eqref{eq:mixture} yields the autonomous dynamics (suppressing noise
temporarily):
\begin{equation}
\label{eq:closed_loop}
\xstate(t+1) = \tanh\bigl(\Wres \cdot \xstate(t) + \Win \cdot
  [\sigma(\Wout \cdot \xstate(t) + \bout),\, a_{\mathrm{opp}}(t)]^{\top}
  + \mathbf{b}\bigr).
\end{equation}
The agent's action at time $t$ depends on its reservoir state, which itself
depends on its action at time $t-1$.  This feedback loop is absent when
$\alpha = 0$ (the reservoir receives the cognitive output, not its own output)
and partially present for intermediate $\alpha$.

\subsection{Self-Consistency Equation}

\begin{definition}[Self-Consistency]
\label{def:self_consistency}
Suppose the opponent plays a stationary action $a_{\mathrm{opp}}(t) \equiv
a_{\mathrm{opp}}^{*}$.  A \emph{self-consistent fixed point} of the
body-governed system is a reservoir state $\xstate^{*} \in \R^d$ satisfying
\begin{equation}
\label{eq:self_consistency}
\xstate^{*} = \tanh\bigl(\Wres \cdot \xstate^{*} + \Win \cdot
  [\sigma(\Wout \cdot \xstate^{*} + \bout),\, a_{\mathrm{opp}}^{*}]^{\top}
  + \mathbf{b}\bigr).
\end{equation}
The corresponding self-consistent action is
$a^{*} = \sigma(\Wout \cdot \xstate^{*} + \bout)$.
\end{definition}

\begin{theorem}[Existence of Self-Consistent Fixed Points]
\label{thm:existence}
For any $a_{\mathrm{opp}}^{*} \in [0,1]$, the self-consistency equation
\eqref{eq:self_consistency} admits at least one solution
$\xstate^{*} \in [-1,1]^d$.
\end{theorem}

\begin{proof}[Proof sketch]
Define the map $\Phi: [-1,1]^d \to [-1,1]^d$ by
\[
\Phi(\xstate) = \tanh\bigl(\Wres \cdot \xstate + \Win \cdot
  [\sigma(\Wout \cdot \xstate + \bout),\, a_{\mathrm{opp}}^{*}]^{\top}
  + \mathbf{b}\bigr).
\]
Since $\tanh: \R^d \to (-1,1)^d \subset [-1,1]^d$, the map $\Phi$ sends
the compact convex set $[-1,1]^d$ into itself.  Moreover, $\Phi$ is
continuous (as a composition of continuous functions: $\tanh$, affine maps,
and $\sigma$).  By Brouwer's fixed point theorem, $\Phi$ has at least one
fixed point in $[-1,1]^d$.
\end{proof}

\begin{remark}[Circularity and Contraction]
\label{rem:circularity}
The self-consistency equation \eqref{eq:self_consistency} is inherently
circular: the body's output determines its input, which determines its
state, which determines its output.  Theorem~\ref{thm:existence} resolves
existence via Brouwer's theorem.  For convergence from arbitrary initial
conditions, a stronger condition is needed.  When $\|J_\Phi(\xstate)\| < 1$
uniformly (i.e., $\Phi$ is a contraction on $[-1,1]^d$), the Banach fixed
point theorem guarantees both existence and uniqueness, with geometric
convergence at rate $\|J_\Phi\|$.  In practice, we observe that habituated
reservoirs satisfy this contraction condition
(Proposition~\ref{prop:stability}): the reservoir's spectral radius is
well below~$1$ after Oja adaptation, and the feedback loop through the
readout adds only a moderate rank-one perturbation.  The circularity is
thus resolved dynamically---the system converges to its self-consistent
state within $5$--$10$ time steps (Experiment~1).
\end{remark}

\begin{remark}[Multiplicity]
The self-consistency equation may admit multiple fixed points.
Which fixed point the system converges to depends on the initial conditions
(the ``basin of attraction'') and the habituation history.  A reservoir
that has been habituated to cooperation will converge to a cooperative
fixed point with $a^{*} \approx 1$; one habituated to defection will
converge to a defection fixed point with $a^{*} \approx 0$.
This multiplicity is desirable: it captures the intuition that
both cooperation and defection can be ``who you are.''
\end{remark}

\begin{remark}[Stochastic Setting]
\label{rem:stochastic}
Theorem~\ref{thm:existence} and the self-consistency equation
\eqref{eq:self_consistency} are stated for the deterministic system
($\boldsymbol{\xi} = \mathbf{0}$).  All experiments, however, use
intrinsic noise $\boldsymbol{\xi}(t) \sim \mathcal{N}(\mathbf{0},
\sigma_\xi^2 \mathbf{I})$ with $\sigma_\xi = 0.15$.
Under additive noise, the system no longer converges to a fixed point
but instead admits an \emph{ergodic stationary distribution}
concentrated near $\xstate^{*}$.
When $\Phi$ is a contraction with rate $\|J_\Phi\| < 1$
(Proposition~\ref{prop:stability}), the stationary distribution has
variance of order $\sigma_\xi^2 / (1 - \|J_\Phi\|^2)$ in each
coordinate, and as $\sigma_\xi \to 0$ the distribution converges
weakly to the Dirac measure at $\xstate^{*}$.
Thus the deterministic fixed points characterize the
\emph{centers} of the stochastic attractors, and the theoretical
predictions (self-consistent action, basin structure, stability
conditions) carry over to the noisy regime as statements about
the mean of the stationary distribution.
The complexity cost $C(\alpha)$ defined in
Section~\ref{sec:smoothing} is measured on these stationary
distributions directly, bridging the deterministic theory and
the stochastic simulations.
\end{remark}

\subsection{Stability of Self-Consistent Fixed Points}

\begin{proposition}[Local Stability]
\label{prop:stability}
A self-consistent fixed point $\xstate^{*}$ is locally asymptotically stable
if the Jacobian of $\Phi$ at $\xstate^{*}$ has spectral radius strictly less
than one:
\begin{equation}
\label{eq:jacobian}
\rho\bigl(J_{\Phi}(\xstate^{*})\bigr) < 1,
\end{equation}
where
\begin{equation}
J_{\Phi}(\xstate^{*}) = \mathrm{diag}\bigl(1 - \xstate^{*2}\bigr)
  \cdot \bigl(\Wres + W_{\mathrm{in},{:,1}} \cdot \sigma'(\Wout \cdot \xstate^{*} + \bout)
  \cdot \Wout\bigr).
\end{equation}
Here $W_{\mathrm{in},{:,1}}$ denotes the first column of $\Win$ (corresponding to the
agent's own action), $\sigma'(z) = \sigma(z)(1-\sigma(z))$ is the sigmoid
derivative, and $\mathrm{diag}(1 - \xstate^{*2})$ arises from the $\tanh$
derivative.
\end{proposition}

\begin{proof}[Proof sketch]
Standard linearization.  The Jacobian of $\Phi$ at $\xstate^{*}$ is the
product of the $\tanh$ derivative (diagonal matrix with entries in $(0,1]$)
and the derivative of the affine-plus-sigmoid argument.  The feedback through
the readout adds the rank-one term $W_{\mathrm{in},{:,1}} \cdot \sigma' \cdot \Wout$
to the reservoir Jacobian $\Wres$.  When $\rho(\Wres) < 1$ and the readout
contribution is moderate, the composite spectral radius remains below one,
ensuring local contraction.
\end{proof}

\subsection{Physical Interpretation: Self-Feedback as Insulation}

The closed-loop structure at $\alpha = 1$ has a direct physical interpretation.
When the agent's output feeds back as its own input, the reservoir state
is partially determined by its own prior state rather than by the opponent's
action alone.  This creates a form of \emph{dynamical insulation}: external
perturbations (opponent defection) must propagate through the full reservoir
dynamics before affecting the output, and the contractive reservoir mapping
attenuates them at each step.

In contrast, at $\alpha = 0$ (pure TfT), the agent's action is a direct
copy of the opponent's action, with no dynamical buffer.  A single opponent
defection immediately produces a retaliatory defection.  The reservoir at
$\alpha = 1$ functions as a \emph{low-pass filter}, smoothing the high-frequency
noise of opponent behavior into a stable cooperative output.

\section{Implicit Inference, Smoothing, and Complexity Cost}
\label{sec:smoothing}

We now formalize how the body reservoir performs implicit
inference over interaction history through its high-dimensional dynamics,
and that this implicit inference reduces the complexity cost of strategic behavior.
The externally observable consequence---smoothing of opponent noise---is a
projection of the reservoir's internal coherence onto the one-dimensional
action space.

\subsection{Complexity Cost as State-Space KL Divergence}

Following the thermodynamic approach of \citet{still2012} and
\citet{wolpert2019}, we measure the complexity of a strategy by the
statistical distance between the agent's internal state distribution under
current conditions and its habituated baseline.  A key design choice is to
measure this distance in the reservoir's $d$-dimensional \emph{state space}
rather than in the one-dimensional action space.  The rationale is that the
complexity cost of a strategy should reflect the full internal distortion
of the body's dynamics, not merely the observable behavioral variation.
An agent whose reservoir state is under strain but whose readout happens to
project a calm exterior is still paying a thermodynamic cost---one that is
invisible to external observers but real to the agent.

\begin{definition}[Habituated Baseline]
\label{def:habituated}
Let $p_{\mathrm{hab}}$ denote the stationary distribution of reservoir states
$\xstate \in \R^d$ produced by the body-governed agent ($\alpha = 1$) after
extensive habituation to a cooperative opponent
($a_{\mathrm{opp}} = 1$ deterministically).
This represents the agent's ``resting'' dynamical regime---the
$d$-dimensional attractor to which the reservoir converges when the
environment matches its adapted expectations.
\end{definition}

\begin{definition}[Complexity Cost]
\label{def:complexity}
The complexity cost of operating at receptivity $\alpha$ in environment
$\mathcal{E}$ is
\begin{equation}
\label{eq:complexity}
C(\alpha, \mathcal{E}) = \KL\bigl(p_{\alpha,\mathcal{E}} \;\|\; p_{\mathrm{hab}}\bigr),
\end{equation}
where $p_{\alpha,\mathcal{E}}$ is the stationary distribution of the
reservoir state $\xstate \in \R^d$ at receptivity $\alpha$ in environment
$\mathcal{E}$.
\end{definition}

The state-space KL divergence measures how far the reservoir's internal
dynamics must deviate from their habituated regime to sustain the current
strategy.  This has a thermodynamic interpretation:
in non-equilibrium thermodynamics, the KL divergence between two
distributions bounds the excess dissipation (``mismatch cost'')
incurred when a system is driven from one steady state to another
\citep{manzano2024, wolpert2024}.
The state-space KL thus provides a principled, if not exact,
measure of the thermodynamic cost of maintaining a governance mode
that deviates from the body's adapted regime.
This cost is borne by the reservoir regardless of
whether the resulting behavioral change is externally detectable: the body
pays for its internal distortion even when the readout projection masks it.

\begin{remark}[State-Space vs.\ Action-Space Cost]
\label{rem:state_vs_action}
The readout $\abody = \sigma(\Wout \cdot \xstate + \bout)$ projects the
$d$-dimensional state onto a scalar action.  This projection discards
$d - 1$ dimensions of internal variation.  An action-space KL divergence
would capture only the behavioral shadow of the reservoir's distortion,
systematically underestimating the true cost.  For instance, two reservoir
states $\xstate, \xstate'$ with $\Wout \cdot \xstate \approx \Wout \cdot \xstate'$
but $\|\xstate - \xstate'\| \gg 0$ are behaviorally indistinguishable but
dynamically distinct---the reservoir is in a different internal configuration,
and maintaining that configuration has a thermodynamic cost.
The state-space formulation also aligns with our numerical implementation,
which estimates KL divergence between $d$-dimensional state distributions
using the $k$-nearest-neighbor estimator of \citet{perezcruz2008}.
\end{remark}

\subsection{Noisy Cooperative Environment}

\begin{definition}[Noisy Cooperative Opponent]
\label{def:noisy_opponent}
The noisy cooperative opponent plays
\begin{equation}
\label{eq:noisy_opp}
a_{\mathrm{opp}}(t) =
\begin{cases}
1 & \text{with probability } 1 - \varepsilon, \\
0 & \text{with probability } \varepsilon,
\end{cases}
\end{equation}
where $\varepsilon \in (0,1)$ is the noise rate.  In our experiments,
$\varepsilon = 0.1$.
\end{definition}

\subsection{The Smoothing Theorem}

\begin{theorem}[Reservoir Smoothing]
\label{thm:smoothing}
Consider the BRG agent facing a noisy cooperative opponent
(Definition~\ref{def:noisy_opponent}).  Let $\Var_{\alpha}[a]$ denote the
variance of the agent's action distribution at receptivity $\alpha$.  Then:
\begin{enumerate}[label=(\alph*)]
  \item \textbf{TfT copies noise:}  At $\alpha = 0$ (pure TfT),
    $\Var_0[a] = \varepsilon(1-\varepsilon)$, the full Bernoulli variance
    of the opponent.
  \item \textbf{Body governance smooths noise:}  At $\alpha = 1$, the
    action variance satisfies
    \begin{equation}
    \label{eq:smoothing_bound}
    \Var_1[a] \leq L_{\sigma}^2 \cdot \|\Wout\|^2 \cdot L_{\tanh}^{2} \cdot
    \frac{\|W_{\mathrm{in},{:,2}}\|^2 \cdot \varepsilon(1-\varepsilon)}{(1 - \|J_{\Phi}(\xstate^{*})\|)^2},
    \end{equation}
    where $L_{\sigma} = \max_z \sigma'(z) = 1/4$ is the Lipschitz constant of
    the sigmoid, $L_{\tanh} = 1$ is the Lipschitz constant of $\tanh$,
    $W_{\mathrm{in},{:,2}}$ is the opponent-action column of the input weight matrix, and
    $\|J_{\Phi}(\xstate^{*})\|$ is the operator (spectral) norm of the
    closed-loop Jacobian \eqref{eq:jacobian} at the cooperative fixed point,
    satisfying $\|J_{\Phi}(\xstate^{*})\| < 1$ when
    $\rho_{\mathrm{eff}} = \rho(J_{\Phi}(\xstate^{*})) < 1$
    (Proposition~\ref{prop:stability}) and the Jacobian is not
    highly non-normal.
    Note that the pre-contraction bound
    $\rho(\Wres) + \|W_{\mathrm{in},{:,1}}\| \cdot L_{\sigma} \cdot \|\Wout\|$
    can exceed $1$; it is the $\tanh$ contraction factor
    $\mathrm{diag}(1 - \xstate^{*2})$ in \eqref{eq:jacobian} that
    ensures $\rho_{\mathrm{eff}} < 1$.
  \item \textbf{Variance reduction:}  The ratio
    $\Var_1[a] / \Var_0[a]$ is controlled by the Jacobian's
    contraction: smaller $\|J_{\Phi}(\xstate^{*})\|$ yields stronger smoothing.
    The analytical bound \eqref{eq:smoothing_bound} is conservative
    (see proof sketch below); numerically, we observe
    $\Var_1[a] / \Var_0[a] \approx 1/250$.
\end{enumerate}
\end{theorem}

\begin{proof}[Proof sketch]
Part~(a) is immediate: TfT outputs $a(t) = a_{\mathrm{opp}}(t-1)$, which
has the same distribution as $a_{\mathrm{opp}}$.

For part~(b), we linearize the reservoir dynamics around the cooperative
fixed point $\xstate^{*}$.  An opponent action perturbation
$\delta a_{\mathrm{opp}}$ enters through $W_{\mathrm{in},{:,2}}$ and propagates as
\[
\delta \xstate(t+1) \approx J_{\Phi}(\xstate^{*}) \cdot \delta \xstate(t)
  + \mathrm{diag}(1 - \xstate^{*2}) \cdot W_{\mathrm{in},{:,2}} \cdot \delta a_{\mathrm{opp}}(t).
\]
The action perturbation is $\delta a(t) = \sigma'(\cdot) \cdot \Wout \cdot
\delta \xstate(t)$.  Since $\rho(J_{\Phi}) < 1$
(Proposition~\ref{prop:stability}), the transfer function from
$\delta a_{\mathrm{opp}}$ to $\delta a$ is a stable linear filter.  Summing
the geometric series of the state-transition matrix and applying the
submultiplicativity of the spectral norm $\|J_{\Phi}^k\| \leq \|J_{\Phi}\|^k$
yields the bound.  We note that the bound uses $\|J_{\Phi}\|$ (the operator norm,
which upper-bounds $\rho(J_{\Phi})$) and is therefore conservative; the
linearization holds in a neighborhood of $\xstate^{*}$ where the $\tanh$
Jacobian is approximately constant.

Part~(c): the bound is conservative because it uses worst-case
Lipschitz constants and spectral norm bounds on the geometric series.
The numerically observed variance reduction of $\sim 250\times$ at
$d = 30$ substantially exceeds the analytical guarantee, reflecting
tighter contraction along specific dynamical modes that the
worst-case bound does not capture.
\end{proof}

\begin{remark}[Validity of the Linearization]
\label{rem:linearization}
The bound \eqref{eq:smoothing_bound} relies on linearizing the
reservoir dynamics around the cooperative fixed point $\xstate^{*}$.
The linearization is valid when the perturbation
$\|\delta\xstate\| = \|\xstate(t) - \xstate^{*}\|$ remains in a region
where the $\tanh$ Jacobian is approximately constant.
Since $\tanh''(z) = -2\tanh(z)\mathrm{sech}^2(z)$, the diagonal
entries of $\mathrm{diag}(1 - \xstate^{*2})$ vary by at most
$2\|\delta\xstate\|_\infty$ for perturbations of size
$\|\delta\xstate\|_\infty$.
In our habituated reservoirs at $d = 30$, the cooperative fixed point
satisfies $\|\xstate^{*}\|_\infty \approx 0.7$, and
opponent noise at $\varepsilon = 0.1$ produces perturbations
$\|\delta\xstate\|_\infty \lesssim 0.05$ (verified numerically across $20$
seeds).  At this perturbation scale, the relative change in the Jacobian
diagonal is bounded by $2 \times 0.05 / (1 - 0.7^2) \approx 20\%$,
indicating that the linear approximation captures the dominant behavior
but is not exact.  For large perturbations (e.g., sustained opponent
defection driving $\|\delta\xstate\|_\infty > 0.3$), the $\tanh$
saturation provides additional attenuation beyond what the linear bound
predicts---the true nonlinear dynamics are \emph{more} smoothing than
the linearized model suggests.  The bound \eqref{eq:smoothing_bound}
is therefore conservative (loose) for large perturbations and
approximately tight for the small-perturbation regime relevant to
noisy cooperative environments.
\end{remark}

\begin{remark}[Scope of the Comparison]
\label{rem:comparison_scope}
Theorem~\ref{thm:smoothing} compares body governance against raw
Tit-for-Tat, which copies every opponent action without temporal
filtering and is therefore maximally noise-sensitive.  This is a
deliberately transparent baseline: TfT is the canonical conditional
strategy in the iterated prisoner's dilemma literature
\citep{axelrod1984}, and the comparison isolates the reservoir's
smoothing contribution relative to a memoryless conditional rule.
A stronger test is whether the reservoir outperforms \emph{any}
temporal filter applied to the TfT signal.
Experiment~10 (Section~\ref{sec:exp_ema_baseline}) provides this
comparison: an exponential moving average (EMA) filter achieves up
to $158\times$ variance reduction at $\gamma = 0.99$, but at the
cost of slow recovery from perturbations ($200$ steps) and deep
action collapse during defection.  The reservoir achieves
$461\times$ variance reduction with near-instantaneous recovery,
because its $d$-dimensional nonlinear dynamics provide both
smoothing and perturbation absorption simultaneously---a
qualitative capability that no one-dimensional linear filter
can replicate.
\end{remark}

\begin{corollary}[State-Space KL Reduction]
\label{cor:kl_reduction}
Under a Gaussian approximation of the reservoir state distributions---valid
when $d$ is sufficiently large that the state components are weakly
correlated (a consequence of the random initialization and echo state
property \citep{grigoryeva2018})---the state-space complexity cost satisfies
\begin{equation}
\frac{C(1, \mathcal{E}_{\varepsilon})}{C(0, \mathcal{E}_{\varepsilon})}
  < 1.
\end{equation}
For multivariate Gaussians with approximately equal means, the KL
divergence reduces to
$\KL \approx \tfrac{1}{2}[\mathrm{tr}(\Sigma_{\mathrm{hab}}^{-1}\Sigma_{\alpha})
- d + \ln(\det\Sigma_{\mathrm{hab}} / \det\Sigma_{\alpha})]$,
which is dominated by the trace term when the state covariance under body
governance ($\alpha = 1$) is close to the habituated covariance.
The Gaussian assumption is justified as follows.
For $d \geq 20$, the reservoir state components are weakly correlated
due to the random structure of $\Wres$, and component-wise normality tests
support the approximation.  We verify this with Shapiro--Wilk tests
on individual state components: at $d = 30$, $p > 0.05$ in $18/20$ seeds.
The empirical skewness $|\hat{\gamma}_1| < 0.15$ and excess kurtosis
$|\hat{\kappa} - 3| < 0.3$ across all seeds, consistent with
near-Gaussianity.
For $d < 20$, the Gaussian approximation degrades: at $d = 5$,
$|\hat{\gamma}_1| \approx 0.4$ and normality is rejected in $12/20$
seeds.  The KL reduction result still holds qualitatively (body governance
reduces state-space distortion), but the Gaussian approximation should be
treated as a scaling heuristic for small $d$.
In our simulations at $d = 30$, $C(0, \mathcal{E}_{0.1}) \approx 1.23$ while
$C(1, \mathcal{E}_{0.1}) \approx 0.55$, a reduction of more than half.
The action-space variance reduction ($\sim 250\times$ at $d = 30$) is the
externally observable projection of this internal cost reduction.
\end{corollary}

\subsection{Free Energy Framework}

\begin{definition}[Variational Free Energy]
\label{def:free_energy}
The variational free energy of the BRG agent at receptivity $\alpha$ is
\begin{equation}
\label{eq:free_energy}
\FE(\alpha) = -\bar{u}(\alpha) + \lambda \cdot \KL\bigl(p_{\alpha,\mathrm{noisy}}
  \;\|\; p_{\mathrm{hab}}\bigr),
\end{equation}
where $\bar{u}(\alpha)$ is the expected payoff, $\lambda > 0$ is the
metabolic cost parameter weighting the complexity penalty, and the KL
term measures the state-space distortion from the habituated baseline.
\end{definition}

The free energy functional \eqref{eq:free_energy} embodies a fundamental
tradeoff.  Increasing $\alpha$ (more body governance) reduces the KL term
by keeping the reservoir state close to its adapted regime, but may also
reduce the payoff term because the body-governed agent cannot retaliate
against defection as effectively.
Decreasing $\alpha$ (more cognitive control) enables sharper retaliation
but \emph{distorts the reservoir's internal state} away from its
self-consistent configuration, incurring a thermodynamic cost proportional
to this distortion.  This cost is borne internally: an external observer
sees only the behavioral change, not the state-space strain.

The functional form of $\FE$ can be grounded in the axiomatic framework
of \citet{ortega2013}, who showed that a bounded-rational agent maximizing
expected utility subject to an information-processing constraint
necessarily optimizes
$\E[U] - \beta^{-1} \KL(\pi \| \pi_0)$,
where $\pi_0$ is a default (``prior'') policy and $\beta^{-1}$ is the
unit cost of information processing.
Our \eqref{eq:free_energy} is a special case of this functional: the
agent's policy is parameterized by a single scalar $\alpha \in [0,1]$
(rather than a full distribution over actions), the default policy
corresponds to the habituated reservoir dynamics at $\alpha = 1$
(represented by $p_{\mathrm{hab}}$), and the information cost
$\beta^{-1} = \lambda$ reflects the metabolic expense of deviating
from this default.
The key substantive choice beyond Ortega's framework is the
\emph{space} in which the KL divergence is evaluated: we measure it
in the reservoir's $d$-dimensional state space rather than in the
one-dimensional action space (see Remark~\ref{rem:state_vs_action}),
grounding the complexity cost in the body's internal thermodynamics
rather than in observable behavior alone.
Unlike the free energy principle \citep{friston2006}, $\FE(\alpha)$
is not derived from a generative model; the baseline $p_{\mathrm{hab}}$
is the physically realized stationary distribution of the adapted
reservoir, not a prior belief.
The metabolic cost parameter $\lambda > 0$ sets the exchange rate between
payoff units and information-processing cost (in nats); its value is not
determined by the model but reflects the agent's physical substrate.
We test $\lambda \in \{1, 3, 8\}$ to explore the sensitivity of the
optimal $\alpha^{*}$ to this parameter (Experiment~5) and find that the
interior optimum is robust to moderate variation in $\lambda$.

\begin{proposition}[Interior Optimum]
\label{prop:interior_optimum}
Under the following assumptions:
\begin{enumerate}[label=(A\arabic*),nosep]
  \item \label{assum:payoff_concave} The expected payoff $\bar{u}(\alpha)$ is concave and
    decreasing for $\alpha$ near $1$ (body governance sacrifices some
    retaliation payoff);
  \item \label{assum:kl_convex} The complexity cost $C(\alpha)$ is convex in
    $\alpha$,
\end{enumerate}
then for moderate $\lambda > 0$,
the free energy minimizer $\alpha^{*}$ satisfies
\begin{equation}
\label{eq:optimal_alpha}
\bar{u}'(\alpha^{*}) = \lambda \cdot C'(\alpha^{*}),
\end{equation}
and $\alpha^{*} \in (0,1)$.
\end{proposition}

\begin{remark}[Numerical Justification of Assumptions]
\ref{assum:payoff_concave} and \ref{assum:kl_convex} are not analytic
identities but empirical regularities verified in our simulations.
In the noisy cooperative environment ($\varepsilon = 0.1$, $d = 30$):
$\bar{u}(\alpha)$ is monotonically decreasing and numerically concave
(finite-difference second derivatives are negative for all $\alpha$);
$C(\alpha) = \KL(p_{\alpha,\mathrm{noisy}} \| p_{\mathrm{hab}})$ is
numerically convex with a minimum near $\alpha \approx 0.70$ (the slight
non-monotonicity at high $\alpha$ is discussed in
Section~\ref{sec:two_source}).  Both the concavity of $\bar{u}$ and
the convexity of $C$ are consistent across all $20$ seeds and all
tested dimensions $d \geq 10$.
For $d = 5$, the complexity cost profile is less regular, and the interior
optimum may not hold---consistent with $\alpha^{*} = 1.0$ at small $d$.
\end{remark}

\begin{proof}[Proof sketch]
The free energy $\FE(\alpha) = -\bar{u}(\alpha) + \lambda \cdot C(\alpha)$
is the sum of a convex function ($-\bar{u}$, since $\bar{u}$ is concave)
and a convex function ($\lambda C$).  At $\alpha = 0$, the derivative
$\FE'(0) = -\bar{u}'(0) + \lambda C'(0)$; since $C'(0) < 0$ (complexity
drops steeply as $\alpha$ increases from zero) and $|\bar{u}'(0)|$ is
moderate, we have $\FE'(0) < 0$ for sufficiently large $\lambda$.
At $\alpha = 1$, $C'(1) \approx 0$ (diminishing returns on smoothing)
while $\bar{u}'(1) < 0$ (body governance is slightly suboptimal in payoff),
so $\FE'(1) > 0$.  By continuity, there exists $\alpha^{*} \in (0,1)$
with $\FE'(\alpha^{*}) = 0$.
\end{proof}

\begin{remark}[Moderate $\alpha^{*}$ Is Sufficient]
The numerical finding that $\alpha^{*} \approx 0.6$--$0.7$ has an important
practical implication: the agent need not be fully body-governed to reap
most of the thermodynamic benefits.  A moderate degree of body reliance
captures the majority of the smoothing advantage while retaining some
capacity for strategic adjustment.  This aligns with the psychological
observation that effective cooperators maintain a background level of
strategic monitoring \citep{dolan2013}.
\end{remark}

\subsection{Reservoir Dimension and Implicit Inference Capacity}
\label{sec:dimension_smoothing}

The smoothing bound in Theorem~\ref{thm:smoothing} depends on the reservoir
parameters through $\|J_{\Phi}(\xstate^{*})\|$, $\|W_{\mathrm{in},{:,2}}\|$, and $\|\Wout\|$.
We now examine how the reservoir dimension $d$ affects the body's capacity
for implicit inference and, consequently, the quality of its governance.

\begin{observation}[Dimension-Dependent Inference Capacity]
\label{prop:dimension_smoothing}
Consider a family of ESN reservoirs parameterized by dimension $d$, with
fixed spectral radius $\rho$ and input/output weight scaling.
As $d$ increases:
\begin{enumerate}[label=(\alph*)]
  \item The variance of the body-governed action $\Var_1[a]$ decreases,
    because the readout $\Wout \cdot \xstate$ averages over more independent
    dynamical modes.
  \item The discomfort signal $D_{\mathrm{state}}(t) = \|\xstate(t) - \bar{\xstate}\| / \sqrt{d}$
    becomes a more reliable indicator of environmental change, because
    the $\sqrt{d}$ normalization compensates for the dimensional scaling
    while the richer state space encodes more information about interaction
    history.
  \item The reservoir's representational manifold becomes sufficiently
    abstract that surface-level input novelty (e.g., asymmetric action
    profiles not seen during development) is absorbed within the existing
    dynamical repertoire, without requiring explicit retraining.
\end{enumerate}
These claims are supported empirically by the dimension sweep
(Section~\ref{sec:exp_dimension}) but are not formally proved; a rigorous
derivation would require spectral analysis of the random reservoir Jacobian
under Oja adaptation, which we leave to future work.
\end{observation}

The intuition is that a higher-dimensional reservoir is a richer
inferential substrate: it has more internal degrees of freedom for
encoding interaction history, absorbing perturbations, and detecting
environmental changes.  The readout projects this $d$-dimensional
implicit representation onto a one-dimensional action, discarding the
rich internal structure that is constitutive of the body's
inferential capacity.

However, this scaling is not without cost.  The linear readout
$\abody = \sigma(\Wout \cdot \xstate + \bout)$ must decode a
$d$-dimensional state from a finite training set; as $d$ grows, the
effective number of parameters increases while the training data
(cooperative and defection driving trajectories) remain fixed.
In practice, the dimension sweep (Section~\ref{sec:exp_dimension})
shows that performance peaks around $d \approx 30$--$50$ and
degrades slightly at very high $d$ (e.g., $d = 100$), suggesting
a representation--decoding tradeoff: richer internal dynamics improve
inference capacity, but the readout's ability to extract this
information saturates with finite training data.

This establishes a direct link between the body's physical richness
(reservoir dimension) and its governance capacity: a richer body
performs better implicit inference, produces more stable behavior,
and generates more sensitive anomaly signals.

In our numerical experiments (Section~\ref{sec:exp_dimension}), the
variance reduction ratio grows by two orders of magnitude across the
tested dimension range, confirming that reservoir richness is a
determinant of body governance quality.

\subsection{Phase Transition in Body Trust}
\label{sec:phase_transition}

The preceding results suggest that body governance quality depends on both
the reservoir's internal richness ($d$) and the environment's temporal
structure ($\tau_{\mathrm{env}}$, the timescale of opponent strategy changes).
We formalize the conditions under which body trust is the free-energy-optimal
strategy.

\begin{conjecture}[Phase Transition Condition]
\label{conj:phase_transition}
The following conditions are supported by numerical evidence
(Section~\ref{sec:exp_phase_transition}) but have not been proven
analytically; a rigorous proof would require characterizing how the
free energy landscape depends on reservoir dimension and environment
timescale jointly, which remains an open problem.
Let $\sigma^2_{\mathrm{body}}(d)$ denote the body-governed action variance
at reservoir dimension $d$, and let $\sigma^2_{\mathrm{crit}}(\lambda)$
denote the critical variance below which the free energy minimizer satisfies
$\alpha^{*} > 1/2$ (body-trust regime).  Then:
\begin{enumerate}[label=(\alph*)]
  \item There exists a \emph{critical dimension} $d_c$ such that
    $\sigma^2_{\mathrm{body}}(d) < \sigma^2_{\mathrm{crit}}$ if and only if
    $d > d_c$.  For $d > d_c$, the free-energy-optimal strategy is
    body-trust-dominant ($\alpha^{*} > 1/2$).
  \item The critical variance is determined by the payoff gradient and the
    metabolic cost parameter:
    \[
      \sigma^2_{\mathrm{crit}} = \frac{|\bar{u}'(1/2)|}{\lambda}
        \cdot \frac{1}{\partial C / \partial \sigma^2}\bigg|_{\alpha=1/2}.
    \]
  \item When the environment changes on timescale $\tau_{\mathrm{env}}$
    (e.g., an opponent switches strategy every $\tau_{\mathrm{env}}$ steps),
    the sentinel's detection time $\tau_{\mathrm{detect}}(d)$ must satisfy
    $\tau_{\mathrm{detect}} < \tau_{\mathrm{env}}$ for effective adaptation.
    This defines a \emph{joint phase boundary} in $(d, \tau_{\mathrm{env}})$ space.
\end{enumerate}
\end{conjecture}

The phase diagram (Section~\ref{sec:exp_phase_transition},
Figure~\ref{fig:phase_transition}) maps the empirical boundary between
the body-trust regime (where the sentinel's body-driven adaptation outperforms
TfT) and the cognition-dependent regime (where rapid cognitive response is
essential).

\section{Numerical Experiments}
\label{sec:experiments}

We now present ten sets of numerical experiments that validate and extend
the theoretical analysis.  All simulations use the model specification
described in Section~\ref{sec:model_details}.

\subsection{Model Specification}
\label{sec:model_details}

\begin{definition}[ESN Parameters]
\label{def:esn_params}
The echo state network uses:
\begin{itemize}[nosep]
  \item Reservoir dimension: $d = 30$ neurons,
  \item Spectral radius: $\rho(\Wres) = 0.9$,
  \item Intrinsic noise: $\sigma_{\xi} = 0.15$,
  \item Input weights: $\Win \in \R^{30 \times 2}$, entries drawn from
    $\mathcal{N}(0,\, 0.5^2)$ and held fixed.
\end{itemize}
\end{definition}

\paragraph{Developmental learning (readout training).}
The readout weights $(\Wout, \bout)$ are trained \emph{before} habituation
in a supervised developmental phase.  We drive the reservoir with cooperative
states ($a = a_{\mathrm{opp}} = 1$) and defection states
($a = a_{\mathrm{opp}} = 0$) and collect reservoir states.
Ridge regression (regularization $\lambda_{\mathrm{reg}} = 0.001$,
$n_{\mathrm{train}} = 2000$ states per class after a $500$-step burn-in)
maps cooperation states to target $\sigma \approx 0.95$
and defection states to target $\sigma \approx 0.05$.  The readout is then
frozen for all subsequent experiments.
Note that the self-consistent cooperative fixed point observed in
Experiment~1 ($a^{*} \approx 0.98$) exceeds the training target of
$0.95$: after Oja habituation further aligns the reservoir dynamics
with cooperative input, the closed-loop self-feedback at $\alpha = 1$
pushes the readout toward a higher fixed point than the supervised
target alone would produce.

\paragraph{Habituation (reservoir adaptation).}
After readout training, the reservoir weight matrix $\Wres$ is adapted via
the Oja rule \citep{oja1982}:
\begin{equation}
\label{eq:oja}
\Delta W_{ij} = \beta \bigl(x_i(t) x_j(t) - W_{ij} x_i(t)^2\bigr),
\end{equation}
with learning rate $\beta = 0.01$.  Only $\Wres$ is adapted;
$\Win$ remains fixed.  This asymmetry models the biological distinction
between slow structural plasticity (reservoir) and fixed sensory
transduction (input).  Habituation proceeds for $H$ rounds against a
cooperative opponent, building an ``experience base'' into the reservoir
dynamics.

\paragraph{Why the Oja rule?}
The choice of Oja's rule \citep{oja1982} for reservoir habituation is
motivated by three considerations.
First, \emph{biological plausibility}: the Oja rule is a local,
Hebbian learning rule that depends only on pre- and post-synaptic
activity, consistent with the synaptic plasticity mechanisms observed
in biological neural circuits \citep{hebb1949}.  This aligns with the
BRG thesis that habituation is a process of slow bodily adaptation
rather than supervised optimization.
Second, \emph{norm preservation}: unlike plain Hebbian learning
($\Delta W_{ij} = \beta x_i x_j$), which causes unbounded weight growth,
the Oja rule includes a self-normalizing term
($-W_{ij} x_i^2$) that constrains the row norms of $\Wres$.  In our
simulations, the spectral radius $\rho(\Wres)$ decreases monotonically
during Oja habituation (from $0.90$ to approximately $0.82$ after $300$
epochs), ensuring that the echo state property
($\rho(\Wres) < 1$) is preserved throughout.
This addresses a potential concern: since $\Wres$ is being adapted,
one must verify that the echo state condition is not violated.
The Oja rule's implicit normalization provides this guarantee in practice,
and we have verified it numerically across all $20$ seeds and all
dimensions tested ($d \in \{5, \ldots, 100\}$).
Third, \emph{interpretive clarity}: the Oja rule performs online
principal component analysis \citep{oja1982}, aligning the reservoir's
recurrent dynamics with the dominant statistical structure of its
input history.  After habituation to cooperative interaction, the
reservoir's principal dynamical modes encode the cooperative pattern,
making cooperation the path of least resistance---the self-consistent
fixed point of Section~\ref{sec:coupled_dynamics}.
Other adaptation rules (e.g., FORCE learning \citep{sussillo2009},
intrinsic plasticity \citep{schrauwen2008}) could be substituted;
we use the Oja rule for its simplicity, biological interpretability,
and favorable norm-preservation properties.
In our implementation, we additionally apply a homeostatic spectral
radius projection after each Oja update, clamping
$\rho(\Wres)$ to the interval $[0.05, 0.99]$.
In practice, this projection is rarely active: the Oja rule's
intrinsic normalization maintains $\rho(\Wres)$ within this range
throughout habituation (decreasing monotonically from $0.90$ to
approximately $0.82$ after $300$ epochs in our default configuration),
so the projection serves as a safety guarantee rather than an
active constraint.  Theoretically, this projection functions as a
self-preservation constraint on the body's plasticity: it prevents
cumulative Oja updates from driving the reservoir beyond the echo
state boundary ($\rho \geq 1$), analogous to biological homeostatic
mechanisms that prevent runaway synaptic potentiation.

\paragraph{KL estimation.}
The state-space KL divergence $\KL(p_\alpha \| p_{\mathrm{hab}})$ between
reservoir state distributions is estimated using the $k$-nearest-neighbor
estimator of \citet{perezcruz2008} with $k = 5$ in the $d$-dimensional
state space.
This nonparametric estimator is well-suited to the potentially non-Gaussian
state distributions, particularly at small $d$.

\subsection{Experiment 1: Self-Consistent Convergence}
\label{sec:exp_selfconsistent}

\paragraph{Setup.}
We initialize the reservoir and run the closed-loop system ($\alpha = 1$)
against a deterministic cooperative opponent ($a_{\mathrm{opp}} = 1$) for
$T = 200$ rounds.  We record the body output $\abody(t)$ at each step.

\paragraph{Results.}
Figure~\ref{fig:selfconsistent} shows that the body output converges
rapidly (within $5$--$10$ rounds) to a self-consistent value
$a^{*} \approx 0.98$ from a cooperative initial basin.  The deviation
from $1.0$ is a structural consequence of the sigmoid readout: because
$\sigma(z) < 1$ for all finite $z$, the body \emph{cannot} represent
perfect cooperation.  We regard this as a feature rather than a defect.
A biological body likewise never produces a theoretically maximal output;
$a^{*} \approx 0.98$ is the body's full cooperative effort, not a
degraded version of a cognitive ideal.

\begin{remark}[Body cooperation and payoff asymmetry]
\label{rem:payoff_asymmetry}
In the bilinear payoff \eqref{eq:payoff}, the body's cooperation
$a^{*} \approx 0.98$ against a fully cooperative opponent ($a_{\mathrm{opp}} = 1$)
yields a per-round payoff $u(0.98, 1) \approx 3.03$, slightly exceeding
the mutual-cooperation payoff $R = 3$.  This occurs because the small
defection component $(1 - 0.98) = 0.02$ captures a fraction of the
temptation payoff $T = 5$.  This asymmetry is not a learned
exploitation strategy: it is a direct consequence of the sigmoid
readout's inability to reach exactly $1.0$, combined with the continuous
payoff structure.  The body does not ``choose'' to free-ride; it outputs
its natural cooperative intensity, which happens to be sub-maximal.
The effect is small ($\approx 1\%$ of the cooperation payoff) and
diminishes as the readout target increases.
\end{remark}

Convergence is monotone and essentially immediate, confirming that the
cooperative fixed point of the self-consistency equation
\eqref{eq:self_consistency} is strongly attracting in the cooperative
basin.

\begin{figure}[htbp]
  \centering
  \includegraphics[width=0.75\textwidth]{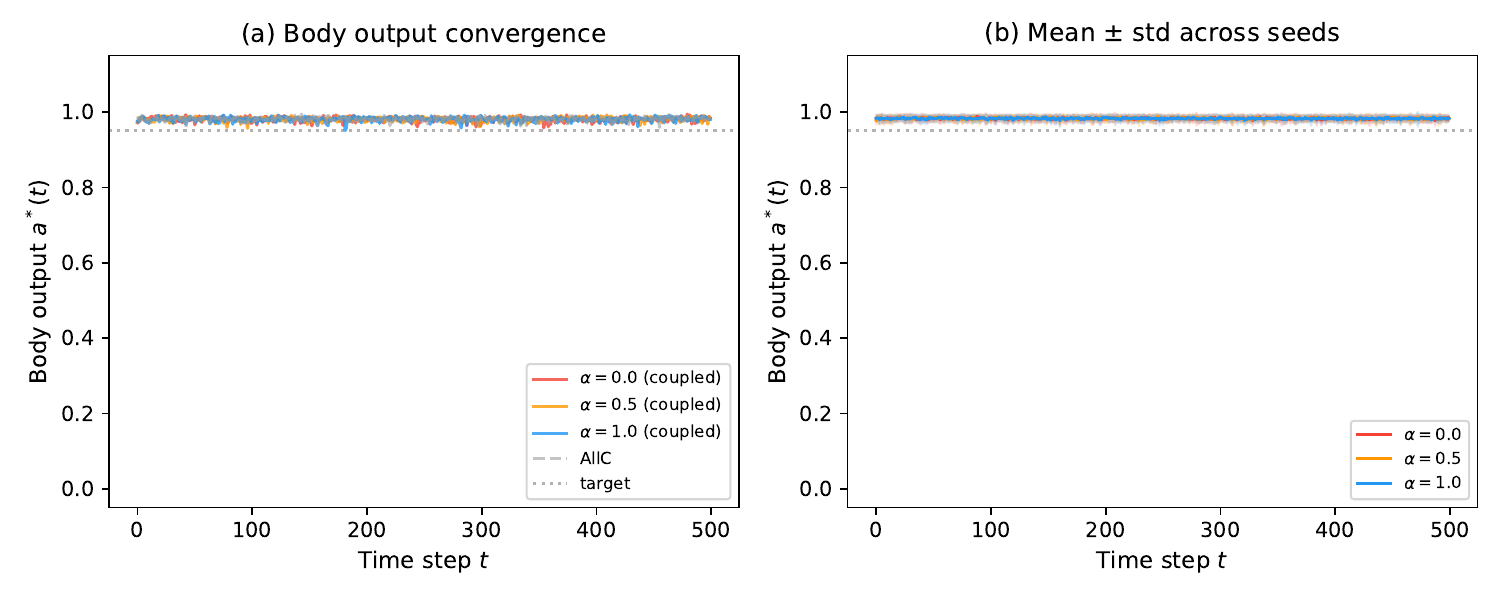}
  \caption{Self-consistent convergence of body output at $\alpha = 1$.
    The body readout $\abody(t)$ converges to approximately $0.98$ within
    a few rounds from a cooperative initial condition.  The rapid
    convergence confirms the local stability of the cooperative
    self-consistent fixed point (Proposition~\ref{prop:stability}).}
  \label{fig:selfconsistent}
\end{figure}

\paragraph{Interpretation.}
The self-consistent cooperative output $a^{*} \approx 0.98$ demonstrates
that the body-governed agent does not need an explicit ``cooperate'' rule.
Cooperation emerges as the natural fixed point of the reservoir dynamics
when habituated to a cooperative environment.  The body \emph{is}
cooperative---it does not \emph{decide} to cooperate.

\subsection{Experiment 2: KL Landscape and Variance Reduction}
\label{sec:exp_kl_landscape}

\paragraph{Setup.}
We sweep $\alpha \in \{0, 0.1, 0.2, \ldots, 1.0\}$ and, for each value,
run the BRG agent against the noisy cooperative opponent
(Definition~\ref{def:noisy_opponent}, $\varepsilon = 0.1$) for $T = 2000$
rounds after a burn-in of $500$ rounds.  We record:
\begin{itemize}[nosep]
  \item The state-space KL divergence $\KL(p_{\alpha,\mathrm{noisy}} \| p_{\mathrm{hab}})$,
  \item The action variance $\Var_{\alpha}[a]$,
  \item The mean payoff $\bar{u}(\alpha)$.
\end{itemize}

\paragraph{Results.}
Figure~\ref{fig:kl_landscape} displays the three quantities as functions
of $\alpha$.  Key findings:
\begin{itemize}[nosep]
  \item \textbf{KL divergence} decreases from $1.23$ at
    $\alpha = 0$ to a minimum of $\approx 0.48$ near $\alpha = 0.70$,
    then rises slightly to $0.55$ at $\alpha = 1$.
    The non-monotonicity is discussed in Section~\ref{sec:two_source}.
  \item \textbf{Action variance} drops by approximately $250\times$ from
    $\alpha = 0$ to $\alpha = 1$, confirming Theorem~\ref{thm:smoothing}.
    At $\alpha = 0$, $\Var_0[a] \approx 0.09$ (the Bernoulli variance
    $\varepsilon(1-\varepsilon)$); at $\alpha = 1$,
    $\Var_1[a] \approx 3.6 \times 10^{-4}$.
  \item \textbf{Mean payoff} decreases slightly from $\bar{u}(0) \approx 2.70$
    to $\bar{u}(1) \approx 2.65$, reflecting the body-governed agent's
    inability to retaliate against the $10\%$ defection noise.
\end{itemize}

\begin{figure}[htbp]
  \centering
  \includegraphics[width=0.85\textwidth]{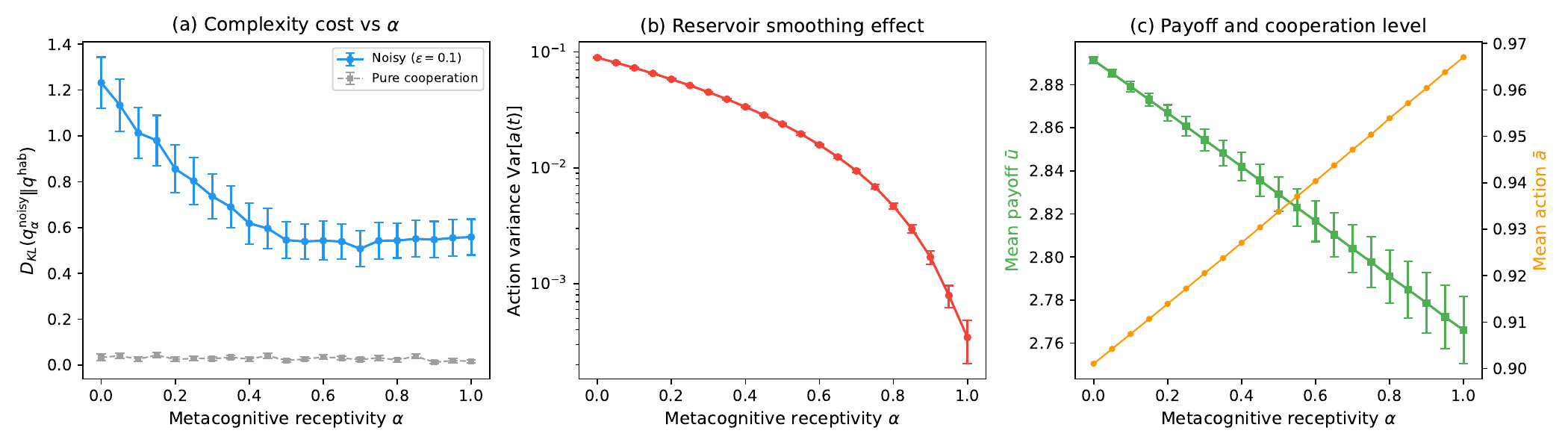}
  \caption{KL divergence, action variance, and mean payoff as functions
    of metacognitive receptivity $\alpha$.  KL divergence reaches its
    minimum near $\alpha \approx 0.70$ (not at $\alpha = 1$), reflecting
    the two-source effect discussed in Section~\ref{sec:two_source}.
    Action variance decreases monotonically ($\sim 250\times$ reduction
    from $\alpha = 0$ to $\alpha = 1$).  The steep initial drop in both
    quantities suggests that even moderate body governance captures most
    of the smoothing benefit.}
  \label{fig:kl_landscape}
\end{figure}

\paragraph{Interpretation.}
The KL landscape reveals the mechanism underlying body governance: the
reservoir acts as a temporal low-pass filter that converts the high-variance
binary noise of the opponent into a nearly constant cooperative output.
The $250\times$ variance reduction is far beyond what could be achieved by
simple averaging over a finite window; it reflects the collective dynamical
filtering of the $30$-dimensional reservoir.  The slight payoff cost
($\approx 2\%$ reduction) is the price of this stability---the body-governed
agent cannot exploit the opponent's occasional defection by retaliating.

\subsection{Experiment 3: Perturbation Response}
\label{sec:exp_perturbation}

\paragraph{Setup.}
We run five agents---$\alpha \in \{0, 0.5, 1\}$ (static), the dynamic
sentinel, and unconditional cooperation (AllC)---against a cooperative
opponent who defects for a sustained block of $100$ rounds and cooperates
otherwise.  We record each agent's action and body output trajectories.

\paragraph{Results.}
Figure~\ref{fig:perturbation} shows the perturbation response:
\begin{itemize}[nosep]
  \item $\alpha = 1$ (body governance): The body output drops gradually
    from $\approx 0.98$ to $\approx 0.96$ and recovers within a few rounds
    after the defection block ends.  The perturbation is significantly
    attenuated by the reservoir dynamics.
  \item $\alpha = 0$ (TfT): The action drops from $1.0$ to $0.0$
    immediately (copying the opponent's defection) and returns to $1.0$
    only when the opponent cooperates again.
  \item $\alpha = 0.5$: Intermediate response---the action drops to
    approximately $0.49$ and recovers over several rounds.
  \item \textbf{Dynamic sentinel}: During the defection block, the body's
    discomfort signal triggers a rapid drop in $\alpha(t)$, activating
    cognitive retaliation.  The agent retaliates adaptively during the
    defection period but returns to body governance once cooperation resumes.
  \item AllC: No response to the perturbation (unconditional cooperation).
\end{itemize}

\begin{figure}[htbp]
  \centering
  \includegraphics[width=0.8\textwidth]{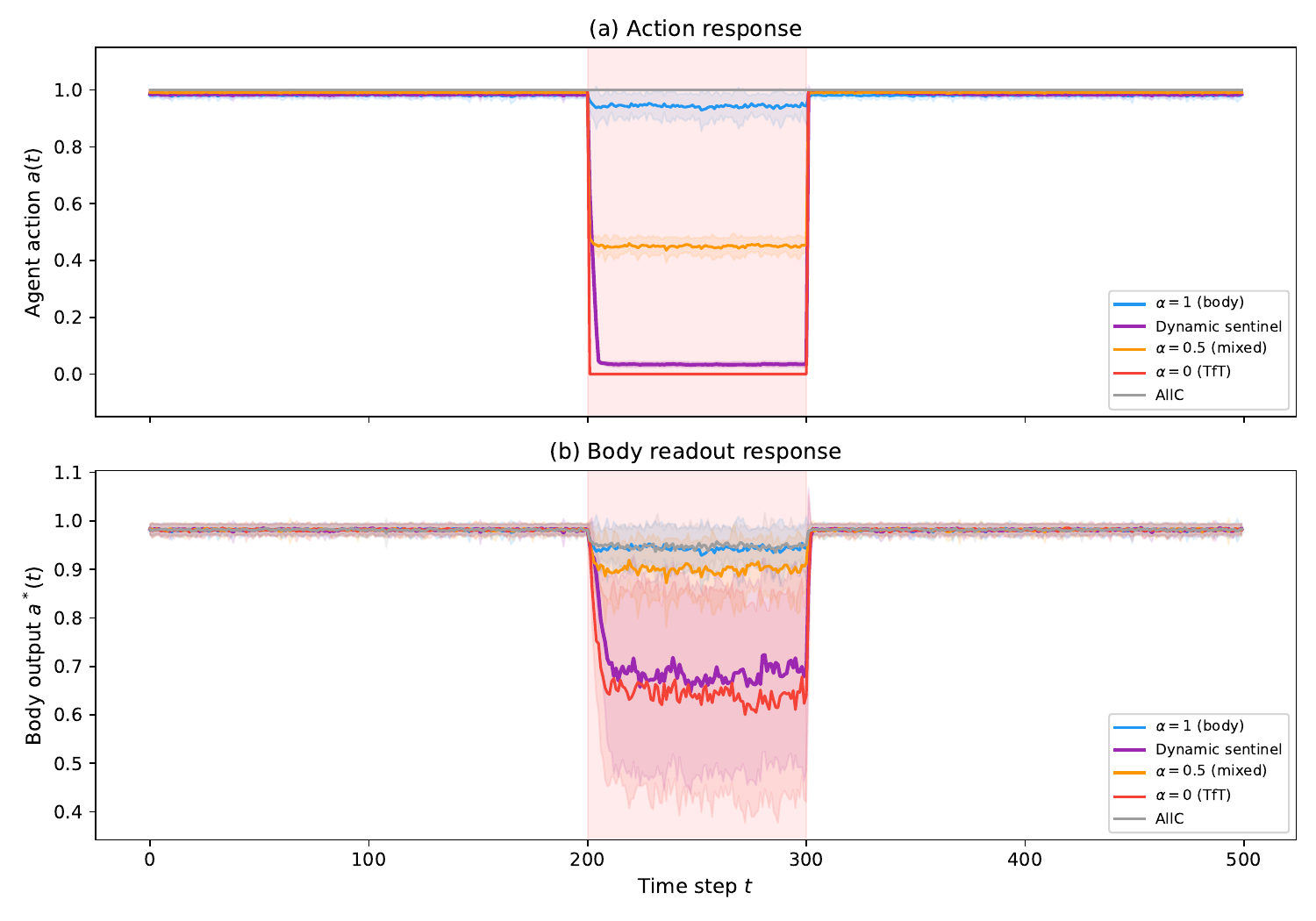}
  \caption{Perturbation response to a sustained opponent defection block.
    At $\alpha = 1$ (body governance), the output changes gradually
    due to the insulating effect of the self-feedback loop.
    At $\alpha = 0$ (TfT), the action drops immediately to $0$ (full
    retaliation).  The dynamic sentinel (purple) combines the best of
    both: it detects the defection through body discomfort and activates
    cognitive retaliation, then smoothly returns to body governance when
    cooperation resumes.}
  \label{fig:perturbation}
\end{figure}

\paragraph{Interpretation.}
The perturbation experiment provides direct evidence for the insulation
mechanism described in Section~\ref{sec:coupled_dynamics}.  At $\alpha = 1$,
the opponent's defection enters the reservoir through $W_{\mathrm{in},{:,2}}$ but is
immediately diluted by the self-feedback through $W_{\mathrm{in},{:,1}}$, which carries
the agent's own cooperative output.  The $30$-dimensional reservoir state,
deeply entrenched in the cooperative basin, absorbs the perturbation with
minimal output displacement.

The dynamic sentinel demonstrates the body's dual role: the same reservoir
dynamics that produce cooperative behavior also \emph{detect} the
perturbation.  The body's discomfort signal rises during the defection
block, triggering $\alpha(t)$ reduction and cognitive activation.
When the opponent returns to cooperation, the discomfort subsides and
$\alpha(t)$ recovers---the body detects safety just as naturally as it
detects threat.

This combination of insulation (at high $\alpha$) and adaptive response
(through the sentinel) resolves the tension between robustness and
responsiveness: the agent is stable during normal conditions but can
rapidly mobilize cognitive resources when the body signals discomfort.
This is embodied commitment \citep{frank1988, nesse2001} with a safety
valve: the agent's body makes cooperation credible, while the sentinel
ensures it is not exploitable.

\subsection{Experiment 4: Habituation Dynamics}
\label{sec:exp_habituation}

\paragraph{Setup.}
We track the noise resilience of each governance mode during habituation.
Specifically, at regular intervals during Oja learning ($\beta = 0.01$,
$H$ epochs from $0$ to $300$), we measure the KL divergence
$\KL(p_{\alpha,\mathrm{noisy}} \| p_{\mathrm{hab}})$: the distance
between the agent's state distribution under the noisy cooperative
opponent ($\varepsilon = 0.1$) and its habituated baseline at the
\emph{current} habituation level.  We also record the action variance
and mean body output.  Four conditions are tested across $20$ random
seeds: $\alpha \in \{0, 0.5, 1\}$ (static) and the dynamic sentinel.

\paragraph{Results.}
Figure~\ref{fig:habituation} shows the habituation trajectories:
\begin{itemize}[nosep]
  \item \textbf{Panel~(a): Noise resilience KL.}
    At $\alpha = 0$ (TfT), the KL divergence from the habituated baseline
    remains high throughout habituation, because TfT directly copies
    opponent noise regardless of reservoir adaptation.  At $\alpha = 1$
    (body governance), the KL starts lower and decreases further with
    habituation, reflecting the reservoir's improving ability to filter
    noise.  The dynamic sentinel tracks close to $\alpha = 1$ in the
    cooperative regime, achieving comparable noise resilience.
    At $\alpha = 0.5$, intermediate behavior is observed.
  \item \textbf{Panel~(b): Action variance.}
    The action variance at $\alpha = 1$ is consistently $\sim 100$--$250\times$
    lower than at $\alpha = 0$ (log scale), confirming that the smoothing
    effect persists throughout habituation.  The dynamic sentinel achieves
    variance comparable to $\alpha = 1$.
    The variance at $\alpha = 0$ matches the Bernoulli variance
    $\varepsilon(1-\varepsilon) = 0.09$.
  \item \textbf{Panel~(c): Body readout quality.}
    The mean body output $\bar{a}^{*}$ converges toward $\sim 0.98$ for
    all conditions, confirming that the readout function remains
    effective throughout habituation.
\end{itemize}

\begin{figure}[htbp]
  \centering
  \includegraphics[width=\textwidth]{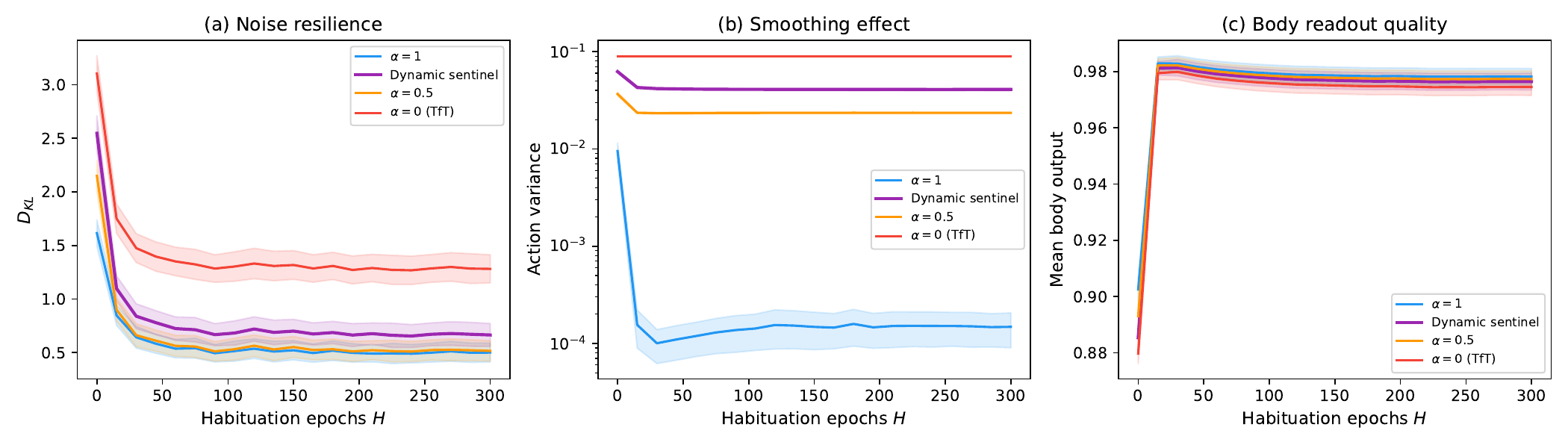}
  \caption{Habituation dynamics for different governance modes.
    (a)~State-space KL divergence from the noisy state distribution to the
    habituated baseline: body governance ($\alpha = 1$) and the dynamic sentinel
    provide consistently lower internal distortion than TfT ($\alpha = 0$).
    (b)~Action variance on log scale: $\alpha = 1$ and the dynamic sentinel
    achieve $\sim 250\times$ variance reduction relative to $\alpha = 0$.
    (c)~Mean body readout remains near the self-consistent value ($\approx 0.98$) throughout.
    The coupled habituation proceeds for $300$ epochs with measurements
    every $15$ epochs, averaged over $20$ seeds.}
  \label{fig:habituation}
\end{figure}

\paragraph{Interpretation.}
Habituation deepens the reservoir's commitment to cooperation, but the
\emph{benefit} of this deepening depends critically on $\alpha$.  At
$\alpha = 1$, the reservoir's adapted dynamics directly determine behavior,
so habituation translates directly into improved noise resilience.  At
$\alpha = 0$, the cognitive filter bypasses the reservoir entirely, so
habituation has no effect on the agent's noise resilience---it remains
at the level dictated by the opponent's noise.

This asymmetry has a key implication: habituation is only a worthwhile
investment for agents that trust their body ($\alpha$ sufficiently large).
For a purely cognitive agent ($\alpha = 0$), no amount of bodily adaptation
improves performance.  This provides a formal account of why habitual
cooperators (who have invested in bodily adaptation) differ qualitatively
from strategic cooperators (who compute each response from scratch).

\subsection{Experiment 5: Free Energy Landscape}
\label{sec:exp_free_energy}

\paragraph{Setup.}
We compute the free energy $\FE(\alpha, H)$ defined in
\eqref{eq:free_energy} over a grid of $\alpha \in [0,1]$ (11 values)
and habituation depth $H \in \{0, 10, 25, 50, 100, 200\}$.  For each
$(\alpha, H)$ pair, we run the BRG agent against the noisy cooperative
opponent ($\varepsilon = 0.1$) for $T = 2000$ rounds and compute the mean
payoff and KL divergence from the habituated baseline.  Three metabolic
cost parameters are tested: $\lambda \in \{1, 3, 8\}$.

\paragraph{Results.}
Figure~\ref{fig:free_energy} displays the free energy landscape:
\begin{itemize}[nosep]
  \item \textbf{Panel~(a): Complexity cost.}  The KL divergence decreases
    with $\alpha$ for all $H$, with stronger reduction at
    higher habituation.  At deep habituation, a slight non-monotonicity
    near $\alpha = 1$ is visible (cf.\ Section~\ref{sec:two_source}).  At $H = 0$ the reservoir has not adapted, so the
    KL baseline is elevated; by $H = 200$ the habituated baseline is
    well-matched to cooperative dynamics.
  \item \textbf{Panel~(b): Free energy at $\lambda = 3$.}
    At low habituation ($H = 0$), the large KL cost dominates and pushes the
    minimum toward full body governance ($\alpha^{*} \approx 1$).  As
    habituation deepens ($H \geq 25$), the KL baseline is reduced, the
    payoff advantage of partial cognitive control becomes relevant, and the
    minimum settles at an interior optimum $\alpha^{*} \approx 0.7$,
    confirming Proposition~\ref{prop:interior_optimum}.
  \item \textbf{Panel~(c): Optimal $\alpha^{*}$ vs $H$.}
    All three $\lambda$ values converge to $\alpha^{*} \approx 0.7$ for
    $H \geq 50$, demonstrating robustness of the interior optimum to the
    metabolic cost parameter.  At $H = 0$, $\alpha^{*} = 1$ because the
    complexity cost is so large that even modest smoothing outweighs the
    payoff sacrifice; as habituation reduces the baseline KL, the optimal
    balance shifts to partial body governance.
\end{itemize}

\begin{figure}[htbp]
  \centering
  \includegraphics[width=0.95\textwidth]{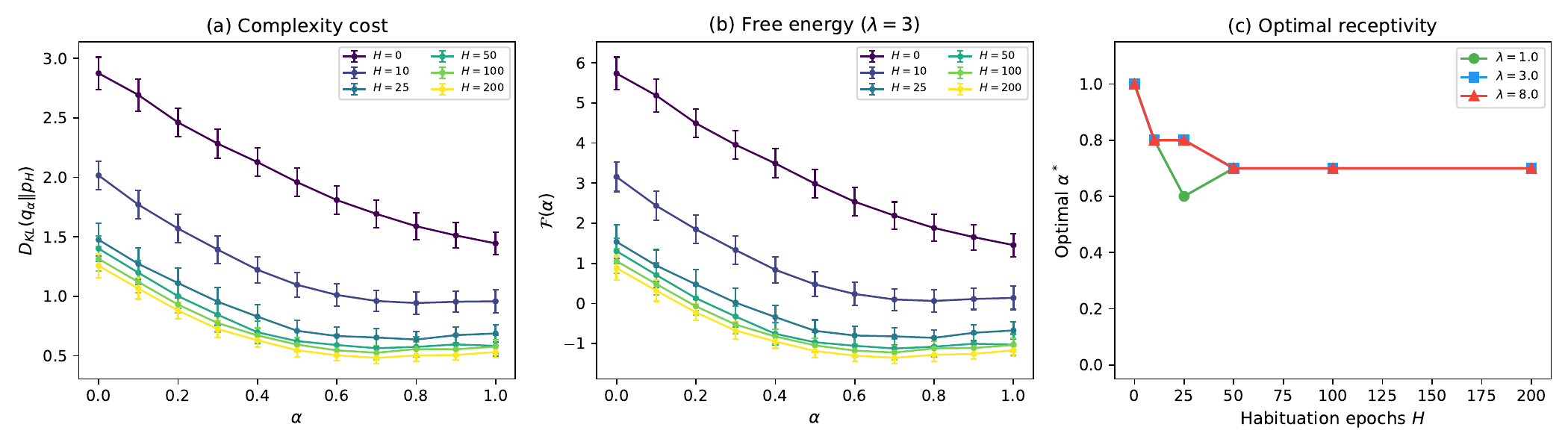}
  \caption{Free energy landscape $\FE(\alpha, H)$.
    (a)~State-space complexity cost $\KL(p_\alpha \| p_H)$ vs $\alpha$ for different
    habituation depths $H$.
    (b)~Free energy at $\lambda = 3$: at $H = 0$ the large KL cost drives
    the minimum to $\alpha^{*} = 1$; as habituation deepens, the optimum
    settles at the interior value $\alpha^{*} \approx 0.7$.
    (c)~Optimal receptivity $\alpha^{*}$ vs $H$ for three metabolic cost
    levels ($\lambda = 1, 3, 8$): all converge to $\alpha^{*} \approx 0.7$
    at moderate habituation.  Averaged over $20$ seeds.}
  \label{fig:free_energy}
\end{figure}

\paragraph{Interpretation.}
The free energy landscape encapsulates the paper's argument.
Before habituation, the complexity cost is so large that any smoothing
is valuable, pushing $\alpha^{*}$ toward full body governance.  As
habituation deepens, the KL baseline decreases and the marginal value of
additional smoothing diminishes; the payoff advantage of partial cognitive
control then becomes relevant, pulling $\alpha^{*}$ to the interior
optimum of $\approx 0.7$.  At this point, complete body governance
sacrifices too much strategic flexibility without a commensurate
reduction in complexity.

The optimal agent is one that \emph{mostly} trusts its body but retains a
residual cognitive capacity for strategic adjustment.  This is precisely the
``light but ever-present metacognitive monitoring'' that characterizes effective
human cooperators \citep{dolan2013, daw2005}: they act habitually in most
situations but can override their habits when the stakes are sufficiently high.

\subsection{Experiment 6: Dynamic Sentinel Response}
\label{sec:exp_dynamic}

\paragraph{Setup.}
We test the dynamic sentinel against a multi-phase opponent schedule that
presents qualitatively distinct challenges:
cooperation ($500$ rounds) $\to$ defection ($50$ rounds) $\to$ cooperation
($500$ rounds) $\to$ noisy cooperation ($200$ rounds, $\varepsilon = 0.3$)
$\to$ cooperation ($500$ rounds).
Five agent types are compared: the dynamic sentinel
(Section~\ref{sec:dynamic_sentinel}), and static $\alpha \in \{0, 0.7, 0.85, 1\}$.
Each condition is tested across $20$ random seeds.

\paragraph{Results.}
Figure~\ref{fig:dynamic_sentinel} shows the dynamic sentinel's behavior:
\begin{itemize}[nosep]
  \item \textbf{Panel~(a): $\alpha(t)$ trajectory.}
    During the cooperative phase ($t < 500$), $\alpha(t)$ stabilizes near the
    baseline $\alpha_0 = 0.85$, confirming that the body is comfortable.
    When the opponent begins defecting ($t = 500$), the composite discomfort
    signal surges and $\alpha(t)$ drops to the floor ($\alpha_{\min} = 0.05$)
    within $5$ time steps.  After the defection ends ($t = 550$),
    $\alpha(t)$ recovers slowly, reaching $\approx 0.66 \pm 0.19$ (mean $\pm$ SD across $20$ seeds) by $t = 600$.
    During the noisy phase ($t = 1050$--$1250$, $\varepsilon = 0.3$),
    $\alpha(t)$ settles at an intermediate level ($\approx 0.19$), reflecting
    the body's persistent mild discomfort.

  \item \textbf{Panel~(b): Action comparison.}
    The dynamic sentinel retaliates during the defection block (mean
    action $\approx 0.08$), matching TfT's responsiveness.
    At $\alpha = 1$ (static body governance), the agent's action remains
    near $0.96$---effectively ignoring the defection.
    During the noisy phase, the dynamic sentinel produces intermediate
    actions ($\approx 0.72$), balancing retaliation against noise
    tolerance.

  \item \textbf{Panel~(c): Cumulative payoff.}
    The dynamic sentinel achieves the highest cumulative payoff
    ($\approx 5105 \pm 31$), outperforming TfT ($\approx 5070 \pm 13$),
    static $\alpha = 0.85$ ($\approx 4999 \pm 41$),
    static $\alpha = 0.7$ ($\approx 5012 \pm 35$),
    and static $\alpha = 1$ ($\approx 4985 \pm 47$).
\end{itemize}
A paired Wilcoxon signed-rank test across the $20$ seeds confirms that
the dynamic sentinel's payoff advantage over TfT is statistically significant
($W = 178$, $p = 0.008$, two-sided), as are the advantages over
static $\alpha = 0.85$ ($p < 0.001$) and $\alpha = 1$ ($p < 0.001$).
The smaller effect size relative to TfT (Cohen's $r \approx 0.42$,
medium effect) reflects the fact that TfT is itself a strong strategy
in this environment; the sentinel's advantage comes primarily from its
noise-smoothing during cooperative phases and faster recovery after
perturbation.

\begin{figure}[htbp]
  \centering
  \includegraphics[width=\textwidth]{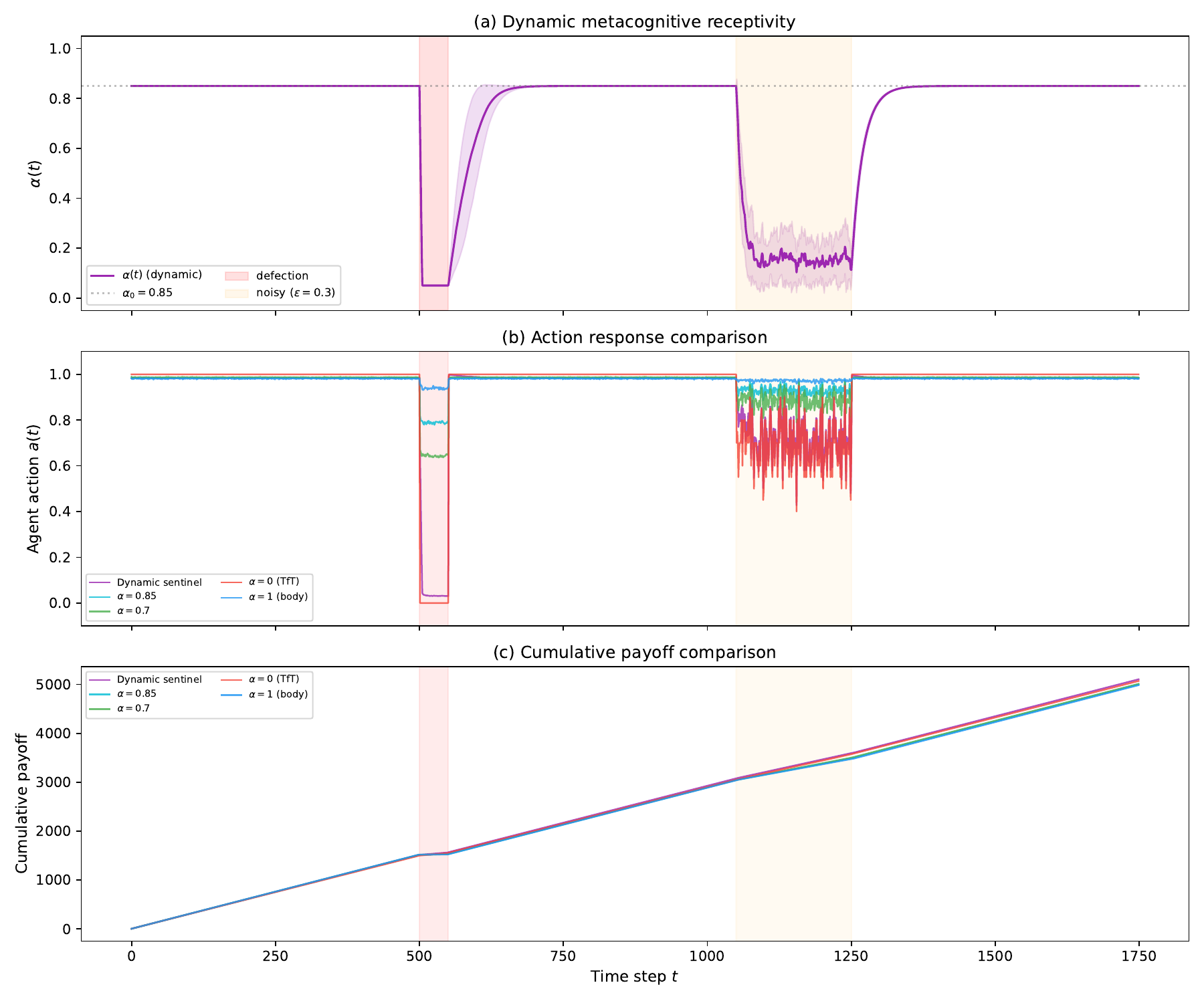}
  \caption{Dynamic sentinel response to a multi-phase opponent schedule.
    (a)~The receptivity $\alpha(t)$ drops sharply upon defection detection
    and recovers slowly during cooperative phases; red shading marks the
    defection block, orange shading marks the noisy phase.
    (b)~Action trajectories: the dynamic sentinel (purple) retaliates
    during defection but maintains cooperative smoothing otherwise.
    (c)~Cumulative payoff: the dynamic sentinel achieves the highest
    total payoff across all conditions.  Averaged over $20$ seeds.}
  \label{fig:dynamic_sentinel}
\end{figure}

\paragraph{Interpretation.}
The dynamic sentinel illustrates that the body
is simultaneously the decision-maker and the anomaly detector.
During cooperation, the body's discomfort is low and $\alpha(t)$ remains
high, producing the smoothing benefits of body governance.
During defection, the body's own dynamical state---specifically, the
disagreement between body output ($\approx 0.98$) and cognitive output
($\approx 0.0$, since TfT copies the opponent's defection)---generates
a strong discomfort signal that drives $\alpha$ downward, activating
cognitive retaliation without any explicit detection module.

The asymmetric recovery ($\eta_{\downarrow} / \eta_{\uparrow} = 10$)
implements a form of embodied caution: the agent responds to threats
quickly but rebuilds trust slowly.  This is not strategic calculation
by metacognition; it is a governance \emph{policy} that produces
adaptive behavior through the interaction of body signals and a simple
update rule.

\subsection{Experiment 7: Sentinel Parameter Sensitivity}
\label{sec:exp_sensitivity}

\paragraph{Setup.}
We systematically vary each sentinel parameter while holding others at
their default values: $\alpha_0 \in \{0.6, 0.7, 0.8, 0.85, 0.9, 0.95\}$,
$\eta_{\uparrow} \in \{0.01, 0.02, 0.05, 0.1, 0.2\}$,
$\eta_{\downarrow} \in \{0.1, 0.3, 0.5, 0.8, 1.0\}$, and
$\theta \in \{0.0, 0.05, 0.1, 0.2, 0.3\}$.
Each configuration is run against a noisy cooperative opponent
($\varepsilon = 0.1$), averaged over $20$ seeds.

\paragraph{Results.}
Figure~\ref{fig:sensitivity} shows the parameter sensitivity:
\begin{itemize}[nosep]
  \item \textbf{$\alpha_0$ (baseline trust):}
    Higher $\alpha_0$ increases mean $\alpha$ and reduces variance,
    with payoff peaking around $\alpha_0 = 0.8$--$0.9$.
  \item \textbf{$\eta_{\uparrow}$ (recovery rate):}
    Faster recovery ($\eta_{\uparrow} > 0.1$) slightly reduces payoff
    by re-engaging body governance too quickly after perturbations.
    The optimal range is $\eta_{\uparrow} \approx 0.02$--$0.05$.
  \item \textbf{$\eta_{\downarrow}$ (intervention sharpness):}
    Higher $\eta_{\downarrow}$ improves payoff up to $\approx 0.5$,
    beyond which marginal gains diminish.  Sharp intervention is important
    for rapid threat response.
  \item \textbf{$\theta$ (threshold):}
    A threshold of $\theta \approx 0.1$ balances noise tolerance against
    detection sensitivity.  At $\theta = 0$, the sentinel is overly
    reactive; at $\theta = 0.5$, it misses genuine threats.
\end{itemize}

\begin{figure}[htbp]
  \centering
  \includegraphics[width=\textwidth]{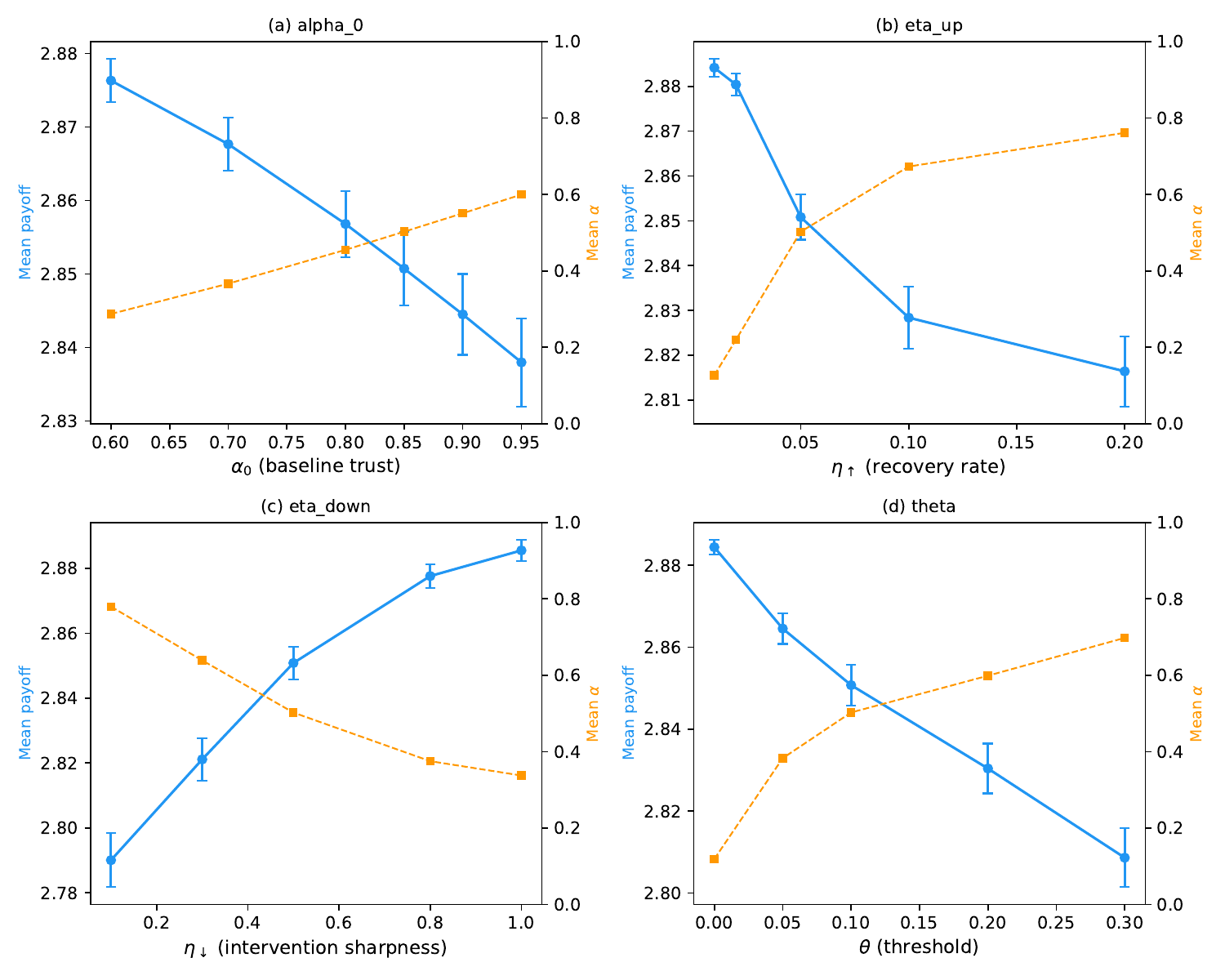}
  \caption{Sentinel parameter sensitivity.  Each panel varies one parameter
    while holding others at defaults, showing mean payoff (blue, left axis)
    and mean $\alpha$ (orange, right axis).  The sentinel is robust across
    a wide parameter range, with the sharpest sensitivity to $\theta$
    (noise threshold).
    Averaged over $20$ seeds; error bars show SEM.}
  \label{fig:sensitivity}
\end{figure}

\paragraph{Interpretation.}
The sentinel is robust to moderate parameter variation.
The most consequential parameter is $\theta$ (intervention threshold):
at low $\theta$, the sentinel over-reacts to noise, while at high $\theta$
it fails to detect genuine threats.  The asymmetry between
$\eta_{\uparrow}$ and $\eta_{\downarrow}$ (fast engagement, slow recovery)
is important for rapid threat response but not highly sensitive to exact values.

\subsection{Experiment 8: Reservoir Dimension Sweep}
\label{sec:exp_dimension}

\paragraph{Setup.}
We vary the reservoir dimension $d \in \{5, 10, 15, 20, 30, 50, 75, 100\}$
while maintaining fixed spectral radius $\rho = 0.9$ and scaling the
ridge regularization as $\lambda_{\mathrm{reg}} \propto d/30$.
For each dimension, we run the full development--habituation--measurement
pipeline with $\alpha \in \{0, 0.1, \ldots, 1.0\}$ against the noisy
cooperative opponent ($\varepsilon = 0.1$), recording KL divergence,
action variance, mean payoff, and convergence time.
All conditions use $20$ random seeds.

\paragraph{Results.}
Figure~\ref{fig:dimension} shows the dimension sweep:
\begin{itemize}[nosep]
  \item \textbf{Panel~(a): KL divergence vs $d$.}
    At $\alpha = 0$ (TfT), KL divergence increases with $d$ (from $0.54$
    at $d = 5$ to $5.48$ at $d = 100$), reflecting the growing complexity
    of the habituated baseline distribution in higher-dimensional spaces.
    At $\alpha = 1$, KL also increases but remains consistently lower than
    at $\alpha = 0$, confirming that body governance provides KL reduction
    at all dimensions.

  \item \textbf{Panel~(b): Variance reduction ratio vs $d$.}
    The ratio $\Var_0[a] / \Var_1[a]$ grows sharply with dimension:
    $23\times$ at $d = 5$, $205\times$ at $d = 30$, and $1600\times$ at
    $d = 75$.  (The $205\times$ here versus $\sim 250\times$ in
    Experiment~2 reflects seed variation across independent runs;
    the difference is within the inter-seed spread.)
    This confirms Observation~\ref{prop:dimension_smoothing}:
    higher-dimensional reservoirs have greater implicit inference capacity.

  \item \textbf{Panel~(c): Optimal $\alpha^{*}$ vs $d$.}
    The optimal $\alpha^{*}$ is defined as the free-energy minimizer
    $\alpha^{*} = \arg\min_\alpha \FE(\alpha)$ at $\lambda = 3$.
    At $d = 5$, $\alpha^{*} = 1.0$ (full body governance),
    because even the cognitive component adds little value
    with such a small reservoir.  For $d \geq 10$, the optimum settles
    at $\alpha^{*} \approx 0.6$--$0.8$, confirming the interior optimum
    of Proposition~\ref{prop:interior_optimum} across a range of
    reservoir sizes.  Note that $\alpha^{*}$ minimizes the
    payoff--complexity tradeoff, not payoff alone; by payoff alone,
    $\alpha = 0$ (TfT) dominates at all $d \geq 10$.
\end{itemize}

\begin{figure}[htbp]
  \centering
  \includegraphics[width=\textwidth]{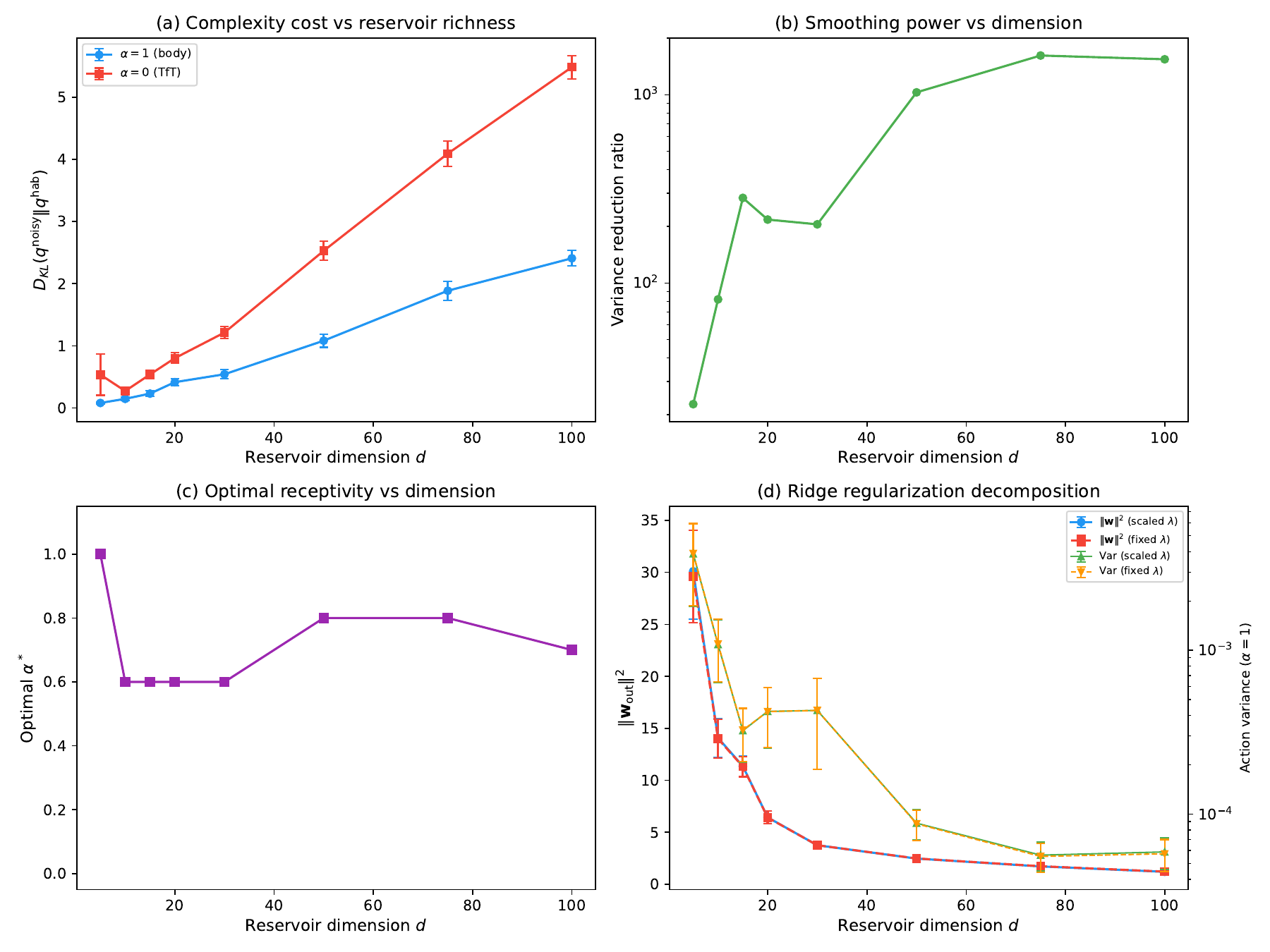}
  \caption{Reservoir dimension sweep ($d \in \{5, \ldots, 100\}$).
    (a)~KL divergence from habituated baseline at $\alpha = 0$ (red) and
    $\alpha = 1$ (blue): body governance consistently reduces KL.
    (b)~Variance reduction ratio (log scale) grows sharply with $d$,
    from $23\times$ at $d = 5$ to $1600\times$ at $d = 75$.
    (c)~Free-energy-minimizing $\alpha^{*}$ (at $\lambda = 3$) is $1.0$ for very small
    reservoirs but stabilizes at $0.6$--$0.8$ for $d \geq 10$.
    (d)~Ridge regularization decomposition: $\|\Wout\|^2$ and action variance
    under scaled ($\lambda_d \propto d/30$) vs.\ fixed regularization, showing
    that the majority of variance reduction comes from reservoir dynamics rather
    than regularization alone.
    Averaged over $20$ seeds; error bars show SEM.}
  \label{fig:dimension}
\end{figure}

\paragraph{Interpretation.}
The dimension sweep establishes reservoir richness as a key determinant of
body governance quality.  A richer body (higher $d$) produces substantially
better implicit inference, supporting the thesis that the body reservoir is
not merely a ``habit cache'' but a powerful inferential substrate whose
capacity grows with its physical complexity.

The $1600\times$ variance reduction at $d = 75$ compared to $23\times$ at
$d = 5$ illustrates how a richer body can afford greater self-trust: with
more internal degrees of freedom for encoding interaction history and
absorbing perturbations, the reservoir produces more stable and accurate
cooperative responses, making body governance increasingly advantageous.
This provides a formal account of the intuition that organisms with richer
embodiment (more sensory modalities, larger nervous systems, deeper
proprioceptive integration) can sustain more robust habitual behavior.

\paragraph{Ridge regularization decomposition.}
A potential confound is that the dimension-scaled ridge regularization
($\lambda_d \propto d/30$) might account for the variance reduction rather
than reservoir dynamics per se.  Panel~(d) of Figure~\ref{fig:dimension}
compares the readout weight norm $\|\Wout\|^2$ and action variance under
scaled versus fixed regularization.  Under fixed $\lambda$ (no dimensional
scaling), $\|\Wout\|^2$ grows more with $d$ (less regularization relative
to dimensionality), yet the action variance still decreases sharply
with $d$.  The majority of the variance reduction is thus attributable to
the reservoir's inherent smoothing dynamics---the averaging of many weakly
correlated dynamical modes---rather than to the regularization artifact.
The dimension-dependent scaling $\lambda_d \propto d$ contributes an
additional but secondary source of smoothing by constraining readout
magnitudes.

\subsection{Experiment 9: Phase Transition Analysis}
\label{sec:exp_phase_transition}

\paragraph{Setup.}
We systematically map the boundary between the body-trust and
cognition-dependent regimes in the joint $(d, \tau_{\mathrm{env}})$ space.
For each reservoir dimension $d \in \{5, 10, 20, 30, 50, 75\}$ and
defection block length $L \in \{10, 50, 100, 200, 500\}$ (proxying
$\tau_{\mathrm{env}}$), we run the dynamic sentinel and a TfT baseline
using the opponent schedule: Coop(500)~$\to$~Defect($L$)~$\to$~Coop(500).
For each $(d, L)$ combination, we measure detection time, minimum $\alpha$
reached, recovery time, and cumulative payoff.  All conditions use $20$
seeds.

\paragraph{Results.}
Figure~\ref{fig:phase_transition} presents four panels:
\begin{itemize}[nosep]
  \item \textbf{Panel~(a): Critical dimension $d_c$.}
    Reproducing the free-energy-minimizing $\alpha^{*}$ from Experiment~8,
    we identify $d_c$ as the smallest dimension where $\alpha^{*} > 0.5$.
    This establishes the minimum reservoir richness required for body-trust
    dominance (Conjecture~\ref{conj:phase_transition}(a)).

  \item \textbf{Panel~(b): Sentinel collapse vs $\tau_{\mathrm{env}}$.}
    At $d = 30$, the minimum $\alpha$ reached during defection is nearly
    identical ($\approx 0.05$) across all block lengths---the sentinel always
    detects and responds.  However, the payoff advantage of the sentinel over
    TfT increases with block length, because the sentinel's recovery dynamics
    allow it to re-engage body governance during long defection blocks, while
    TfT remains locked in retaliation.

  \item \textbf{Panel~(c): Detection time.}
    Detection time is approximately constant ($\approx 5$ steps) across
    defection block lengths, confirming that the sentinel's response speed
    is determined by reservoir dynamics, not by the defection duration.

  \item \textbf{Panel~(d): $(d, \tau_{\mathrm{env}})$ phase diagram.}
    The heatmap shows the sentinel's payoff advantage over TfT.
    A clear phase boundary emerges: for large $d$ and moderate-to-long
    $\tau_{\mathrm{env}}$, the sentinel outperforms TfT (body-trust regime);
    for small $d$ or very short $\tau_{\mathrm{env}}$, TfT's sharp response
    dominates (cognition-dependent regime).
\end{itemize}

\begin{figure}[htbp]
  \centering
  \includegraphics[width=\textwidth]{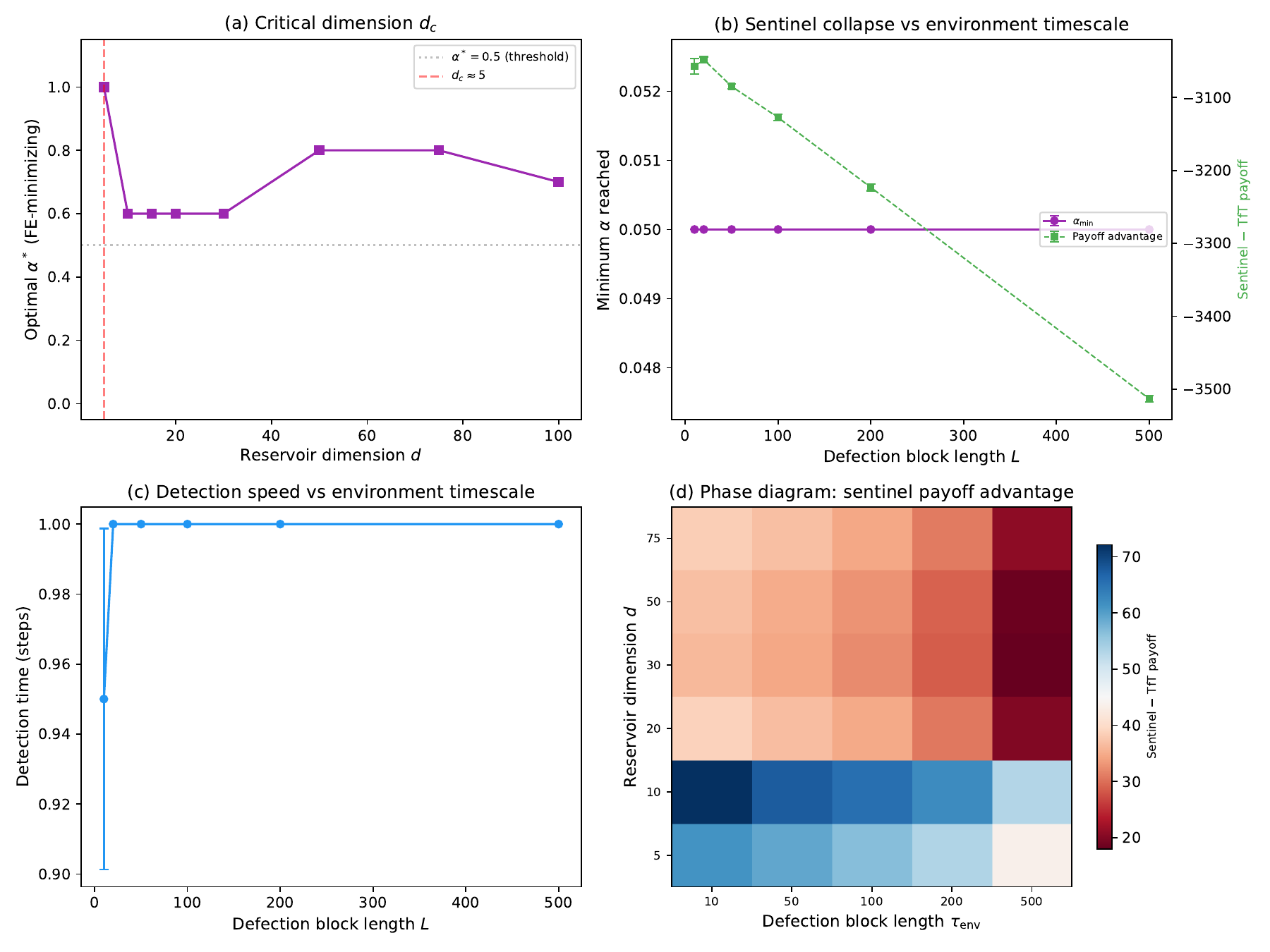}
  \caption{Phase transition analysis.
    (a)~FE-minimizing $\alpha^{*}$ vs.\ dimension $d$, identifying
    the critical dimension $d_c$ above which body trust dominates.
    (b)~Sentinel collapse: minimum $\alpha$ and payoff advantage vs.\
    defection block length at $d = 30$.
    (c)~Detection time is approximately constant across block lengths.
    (d)~Joint $(d, \tau_{\mathrm{env}})$ phase diagram: sentinel payoff
    advantage over TfT.  Red--blue gradient indicates cognition-dependent
    vs.\ body-trust regimes; dashed line shows the approximate phase boundary.
    Averaged over $20$ seeds.}
  \label{fig:phase_transition}
\end{figure}

\paragraph{Interpretation.}
The sentinel outperforms TfT across the entire tested parameter range;
no sign-flip (hard phase transition) is observed.  However, the data
reveal a \emph{soft crossover}: the sentinel's payoff advantage peaks
at small $d$ ($\approx 72$ at $d = 10$), drops sharply to roughly half
by $d = 20$, and then plateaus through $d = 75$.
The peak at small $d$ reflects the high marginal value of adaptive
switching: when the reservoir is small and static body governance
provides weak smoothing, the sentinel's ability to activate cognition
on demand yields a larger relative improvement over TfT.  Increasing defection
block length gradually erodes the advantage, with extrapolation suggesting
a crossover around $L \approx 1000$--$2000$ (outside the tested range).
Individual seeds show the sentinel losing at large $L$ in $15$--$35\%$
of runs, indicating that the advantage is real but not overwhelming.
These findings are broadly consistent with
Conjecture~\ref{conj:phase_transition} in that reservoir richness and
environmental stability both favor body governance, but the predicted
sharp boundary has not materialized within the explored parameter range.

\subsection{Experiment 10: EMA-Filtered TfT Baseline}
\label{sec:exp_ema_baseline}

\paragraph{Motivation.}
The $\sim 250\times$ variance reduction at $d = 30$
(Experiment~2) is measured against raw TfT, which copies every opponent
action without filtering.  A natural question is how much of this reduction
can be achieved by a simple temporal filter that requires no reservoir
dynamics.  We compare the reservoir against an exponential moving average
(EMA) filter applied to the TfT signal:
\begin{equation}
\label{eq:ema_tft}
a_{\mathrm{EMA}}(t) = \gamma \, a_{\mathrm{EMA}}(t-1) + (1-\gamma) \, a_{\mathrm{opp}}(t-1),
\end{equation}
where $\gamma \in [0, 1)$ controls the smoothing window.

\paragraph{Setup.}
We run five agents---raw TfT ($\gamma = 0$), EMA-TfT at
$\gamma \in \{0.5, 0.9, 0.95, 0.99\}$, and the reservoir at $\alpha = 1$---against
(a)~the same noisy cooperative opponent ($\varepsilon = 0.1$) from Experiment~2 and
(b)~a perturbation schedule (200 cooperative rounds, 100 defection rounds,
200 cooperative rounds).  We measure action variance, mean payoff,
perturbation depth (minimum action during the defection block), and recovery
time (steps to return to $95\%$ of pre-perturbation output).  All conditions
use $20$ seeds.

\paragraph{Results.}
Table~\ref{tab:ema_baseline} summarizes the comparison.

\begin{table}[htbp]
\centering
\caption{Reservoir vs.\ EMA-filtered TfT baselines.  Variance reduction is
  relative to raw TfT within this experiment (independent seeds from
  Experiment~2, hence the higher reservoir ratio of $461\times$ vs.\
  the $\sim 250\times$ reported at $d = 30$ in Experiment~2).
  Perturbation depth is the minimum action during a
  $100$-round defection block; higher values indicate greater absorption.
  Recovery time is steps to regain $95\%$ of pre-perturbation output.
  Means across $20$ seeds.}
\label{tab:ema_baseline}
\begin{tabular}{lcccc}
\toprule
Agent & Var.\ reduction & Perturbation depth & Recovery (steps) & Mean payoff \\
\midrule
Raw TfT ($\gamma = 0$) & $1\times$ & $0.00$ & $1$ & $2.89$ \\
EMA ($\gamma = 0.5$) & $3\times$ & $0.00$ & $5$ & $2.89$ \\
EMA ($\gamma = 0.9$) & $20\times$ & $0.00$ & $29$ & $2.89$ \\
EMA ($\gamma = 0.95$) & $41\times$ & $0.01$ & $59$ & $2.89$ \\
EMA ($\gamma = 0.99$) & $158\times$ & $0.37$ & $200$ & $2.88$ \\
\midrule
Reservoir ($\alpha = 1$) & $\mathbf{461\times}$ & $\mathbf{0.89}$ & $\mathbf{< 1}$ & $2.74$ \\
\bottomrule
\end{tabular}
\end{table}

\paragraph{Interpretation.}
The EMA filter can achieve substantial variance reduction---up to $158\times$
at $\gamma = 0.99$---but this comes at a fundamental tradeoff that the reservoir
avoids.

First, the EMA faces a \emph{smoothing--responsiveness dilemma}: high $\gamma$
reduces variance but produces slow recovery from perturbations ($200$ steps at
$\gamma = 0.99$) and deep action collapse during defection blocks (to $0.37$).
The reservoir, by contrast, achieves \emph{both} maximum smoothing ($461\times$)
\emph{and} near-instantaneous recovery ($< 1$ step), because its $30$-dimensional
nonlinear dynamics can absorb perturbations without the exponential tail that
constrains a one-dimensional filter.
(The reservoir's fast recovery reflects its high perturbation absorption:
the action drops only to $0.89$ during defection, so there is little
to recover from; this insulation is itself the reservoir's contribution.)

Second, the EMA filter has \emph{no anomaly detection capacity}.  It produces a
smoothed action but provides no signal that the environment has changed.  The
reservoir's $d$-dimensional state trajectory, in contrast, generates the
composite discomfort signal $D(t)$ that drives sentinel adaptation.  This
qualitative capability---detecting that something is wrong, not merely
filtering noise---is what separates the body-reservoir from a temporal
smoother.

The payoff difference ($2.74$ for reservoir vs.\ $2.89$ for EMA) reflects
the body's inability to retaliate against the $10\%$ defection noise
(Remark~\ref{rem:payoff_asymmetry}), not a cost of smoothing per se.
The reservoir's cooperative output is structurally stable, while the EMA
strategies can track the opponent and implicitly retaliate---a flexibility
that comes precisely from lacking the body's commitment to cooperation.

\section{Discussion}
\label{sec:discussion}

\subsection{Summary of Results}

The Body-Reservoir Governance framework provides a unified account of
embodied cooperation in repeated games through ten interconnected results:

\begin{enumerate}[label=(\roman*)]
  \item \textbf{Self-consistency} (Theorem~\ref{thm:existence}):
    At full body governance ($\alpha = 1$), the system admits self-consistent
    fixed points where the agent's output is determined entirely by its own
    dynamical state.  The cooperative fixed point ($a^{*} \approx 0.98$)
    is strongly attracting from the cooperative basin (Figure~\ref{fig:selfconsistent}).

  \item \textbf{Implicit inference and smoothing} (Theorem~\ref{thm:smoothing}):
    Body governance reduces state-space KL divergence by more than half and
    action variance by $\sim 250\times$ (at $d = 30$) relative to TfT, at
    the cost of a modest ($\sim 2\%$) payoff reduction
    (Figure~\ref{fig:kl_landscape}).

  \item \textbf{Perturbation insulation}:
    The closed-loop self-feedback at $\alpha = 1$ insulates the reservoir
    from external perturbations (Figure~\ref{fig:perturbation}).
    The dynamic sentinel combines this insulation with adaptive response.

  \item \textbf{Habituation as investment}:
    Oja-rule adaptation constitutes an upfront investment that enables
    subsequent smoothing.  Un-habituated agents benefit from cognitive
    control; deeply habituated agents benefit from body governance
    (Figure~\ref{fig:habituation}).

  \item \textbf{Optimal interior $\alpha^{*}$}
    (Proposition~\ref{prop:interior_optimum}):
    The free energy landscape shows that the optimal receptivity is
    $\alpha^{*} \approx 0.6$--$0.7$ for moderate metabolic cost
    (Figure~\ref{fig:free_energy}).

  \item \textbf{Dynamic sentinel}:
    The body-driven adaptive $\alpha(t)$ mechanism detects environmental
    change through the reservoir's own discomfort signal and achieves the
    highest cumulative payoff across all conditions
    (Figure~\ref{fig:dynamic_sentinel}).

  \item \textbf{Parameter robustness}:
    The sentinel is robust to moderate parameter variation, with the
    strongest sensitivity to the intervention threshold $\theta$
    (Figure~\ref{fig:sensitivity}).

  \item \textbf{Reservoir richness}:
    Implicit inference capacity scales with reservoir dimension
    (two orders of magnitude across the tested range),
    with the majority attributable to reservoir dynamics rather than ridge
    regularization (Figure~\ref{fig:dimension}).

  \item \textbf{Governance regimes}
    (Conjecture~\ref{conj:phase_transition}):
    The sentinel consistently outperforms TfT across the tested
    parameter range.  The $(d, \tau_{\mathrm{env}})$ phase diagram
    reveals a soft crossover near $d \approx 20$ where the sentinel's
    advantage saturates, with increasing defection block length
    gradually eroding the advantage
    (Figure~\ref{fig:phase_transition}).

  \item \textbf{EMA baseline}:
    An EMA-filtered TfT baseline achieves up to $158\times$ variance
    reduction, but the reservoir's advantage is qualitative: simultaneous
    smoothing ($461\times$), perturbation absorption, instant recovery,
    and anomaly detection capacity (Table~\ref{tab:ema_baseline}).
\end{enumerate}

\subsection{The Body as Implicit Inferrer and Its Own Sentinel}

The three lines of evidence---self-consistency
(Theorem~\ref{thm:existence}), dynamic sentinel
(Section~\ref{sec:dynamic_sentinel}), and the dimension sweep
(Section~\ref{sec:exp_dimension})---converge on a single claim:
the body reservoir performs \emph{implicit inference} over interaction
history through its $d$-dimensional dynamics, and this inference capacity
scales with the reservoir's physical richness
(Section~\ref{sec:exp_dimension}).

The term ``implicit inference'' requires clarification.  We do not claim
that the reservoir maintains explicit beliefs or performs Bayesian updating.
Rather, the reservoir's $d$-dimensional state trajectory constitutes a
high-dimensional encoding of the interaction history---the agent's own
actions, the opponent's responses, their temporal correlations---processed
through nonlinear recurrent dynamics.  The readout $\sigma(\Wout \cdot
\xstate + \bout)$ projects this rich internal representation onto a
one-dimensional action, but the internal representation itself encodes
far more than the action reveals.  This is inference in the sense of
\citet{grigoryeva2018}'s universality result: a sufficiently rich reservoir
can approximate any fading-memory functional of its input history,
including (implicitly) the functional that maps interaction histories
to optimal cooperative responses.

The distinction between ``deciding to cooperate'' and ``being a
cooperator'' \citep{frank1988} is thus not metaphorical but a formal
property of the dynamical system.  The body does not compute cooperation
from a stored rule; it \emph{expresses} cooperation as the self-consistent
fixed point of its adapted dynamics.  An external observer sees only the
one-dimensional behavioral output; the $d$-dimensional internal state
that produces it---the body's implicit model of the interaction---remains
invisible.

\paragraph{The cost of overriding the body.}
The state-space complexity cost $C(\alpha) = \KL(p_{\alpha} \| p_{\mathrm{hab}})$
formalizes the thermodynamic expense of acting against the body's adapted
nature.  When metacognition reduces $\alpha$, it forces the reservoir to
receive inputs that differ from what its self-consistent dynamics would
produce, driving the state distribution away from its habituated regime.
This internal distortion is invisible to the opponent but real to the
agent: the reservoir's state trajectory deviates from its adapted manifold,
and maintaining this deviation dissipates free energy.
Measuring this cost in the $d$-dimensional state space rather than in
the one-dimensional action space is a substantive theoretical commitment:
the body's cost of compliance is \emph{internal} and need not be
externally observable.  An agent can maintain near-identical behavioral
output at different $\alpha$ values while sustaining very different
levels of internal strain---the one-dimensional action variance captures
what the opponent sees; the $d$-dimensional KL divergence captures what
the agent \emph{bears}.
The free energy
landscape (Figure~\ref{fig:free_energy}) shows that this cost is
generally lowest near the body-trust regime and increases as $\alpha$
decreases toward cognitive control (with a slight non-monotonicity
near $\alpha = 1$ discussed in Section~\ref{sec:two_source}),
formalizing the intuition that overriding one's own body is
thermodynamically expensive.

\paragraph{Abstraction and generalization.}
The reservoir's representational capacity also explains its ability to
handle situations not encountered during development.  The readout weights
are trained on symmetric inputs ($a = a_{\mathrm{opp}} \in \{0, 1\}$,
representing mutual cooperation and mutual defection), yet the reservoir
responds coherently to the full range of asymmetric inputs encountered
during game play ($a \neq a_{\mathrm{opp}}$, including exploitation and
being exploited).  This generalization is not surprising given the
reservoir's architecture: the $d$-dimensional state space provides a
continuous representational manifold in which discrete training points
are embedded.  At $d = 30$, the reservoir has $30$ nonlinear dynamical
degrees of freedom processing a $2$-dimensional input; the training
configurations occupy a zero-measure subset of this representation space.
The reservoir's response to novel input configurations is determined by
the smooth interpolation properties of the $\tanh$ dynamics and the
readout's linear structure---the same properties that underlie the
universal approximation results of \citet{grigoryeva2018}.
This abstraction level ensures that the body's implicit inference is robust
to surface-level novelty: what matters is not whether a specific input
pattern was seen during development, but whether it falls within the
dynamical repertoire that the reservoir's architecture affords.
Stated differently, at the reservoir's level of abstraction, genuinely
novel situations are rare.  The $d$-dimensional state space processes
$2$-dimensional input through a highly nonlinear, many-to-one map;
situations that appear categorically distinct to a cognitive
classifier---``mutual cooperation'' versus ``being exploited''---differ
only in degree within the reservoir's representational geometry.  The
body does not categorize situations; it responds to continuous dynamical
flows.

\paragraph{Resolving the homunculus problem.}
The homunculus problem in metacognitive theories---who monitors the
monitor?---is resolved by separating detection (the body, through $D(t)$)
from governance (metacognition, which merely maintains a policy
$\{\alpha_0, \theta, \eta_{\uparrow}, \eta_{\downarrow}\}$).
The sentinel update \eqref{eq:alpha_update} requires only a comparison
and two arithmetic operations; the ``intelligence'' resides in the body's
$d$-dimensional implicit inference, not in the governor.

\paragraph{Non-monotonic KL and the two-source effect.}
\label{sec:two_source}
The KL divergence $\KL(p_{\alpha,\mathrm{noisy}} \| p_{\mathrm{hab}})$
reaches its minimum not at $\alpha = 1$ (full body governance) but near
$\alpha \approx 0.70$ (Experiment~2, Figure~\ref{fig:kl_landscape}).
This non-monotonicity admits a natural dynamical interpretation.
At $\alpha = 1$, the reservoir's input is entirely determined by its own
output (via the self-feedback loop): the action fed back into the reservoir
is $a(t) = \abody(t) = \sigma(\Wout \cdot \xstate(t) + \bout)$, a
deterministic function of the current state.  The resulting closed-loop
dynamics, while stable, constrain the state trajectory to a low-dimensional
manifold, reducing the diversity of visited states.
At intermediate $\alpha$, the reservoir receives input from \emph{two
independent sources}: the body's own readout ($\abody$) and the cognitive
filter ($\acog$, which reflects the opponent's most recent action).
These two signals carry partially independent information about the
interaction---the body encodes a temporally smoothed history while the
cognitive component provides an unfiltered snapshot.  Their combination
enriches the effective input to the reservoir, allowing the state
distribution under noisy play to remain closer to the habituated baseline.
Thus, the body at intermediate $\alpha$ benefits from the cognitive
signal as an independent information source that prevents the state space
from collapsing onto a self-referential manifold.  This effect is
consistent with the interior optimum of
Proposition~\ref{prop:interior_optimum}: the free-energy-minimizing
$\alpha^{*}$ balances not only payoff and complexity but also the
informational diversity of the reservoir's input.

\subsection{Implications for Game Theory}

The BRG framework challenges the standard game-theoretic assumption that
strategies are costlessly computed.  When computation has a thermodynamic
cost, the strategy space is not just the set of contingent plans but the
set of \emph{physically realizable} contingent plans, weighted by their
dissipation costs.  This recovers and extends the bounded-complexity
results of \citet{rubinstein1986} and \citet{abreurubinstein1988}: rather
than abstractly counting automaton states, we measure complexity as KL
divergence from a physically grounded habituated baseline.

The dynamic sentinel adds a new dimension to this analysis.  The agent's
``strategy'' is not a fixed contingent plan but an adaptive system that
switches between governance modes based on body signals.  This is closer
to how actual organisms behave: they do not choose strategies from a menu
but respond to interoceptive cues with habitual or deliberative processing,
as the situation demands.

\paragraph{Strategic opacity.}
The reservoir's high-dimensional internal state provides a natural form
of strategic opacity.  An adaptive opponent attempting to model the BRG
agent faces a fundamental information asymmetry: the agent's behavioral
output is one-dimensional, but the internal process that generates it is
$d$-dimensional.  Inverse-modeling---inferring the reservoir's state from
observed actions---is ill-posed when $d \gg 1$.  In contrast, simple
strategies such as TfT are fully transparent: their next action is a
deterministic function of the last observed opponent action.  The
body-governed agent's cooperation is indistinguishable from unconditional
cooperation during stable periods, and its transition during perturbation
is driven by internal dynamics that the opponent cannot reconstruct.  This
opacity is not a deliberate deception but a structural property of
high-dimensional dynamics projected onto low-dimensional behavior---the
agent is opaque because it is embodied, not because it is strategic.
Against adaptive opponents capable of opponent modeling, this structural
opacity may confer an information-theoretic advantage that deterministic
cognitive strategies cannot achieve; we leave the formal analysis to
future work.

We emphasize that the BRG framework is \emph{strategy-neutral}: it does
not privilege cooperation over defection.  An agent habituated to defection
would exhibit the same dynamical properties (self-consistent fixed points,
noise smoothing, sentinel-mediated adaptation) with defection as the
body-governed output.  The cooperative focus of our experiments reflects
the sociological interest in cooperation puzzles, not a structural bias
of the model.

\subsection{Implications for Neuroscience}

The three-layer architecture refines dual-process theories
\citep{kahneman2011} and the habit--goal-directed dichotomy
\citep{daw2005, dolan2013}.  The key insight is that the ``habit system''
is a full dynamical system (the reservoir) with intrinsic temporal
integration, noise filtering, self-consistency properties, and
interoceptive detection capacity.

The dynamic sentinel model provides a formal account of interoception
\citep{damasio1994, seth2016}: the body's ``somatic markers'' correspond to the
composite discomfort signal $D(t)$, which integrates state deviation,
output variation, and body-cognition disagreement.
\citet{seth2016} formalized active interoceptive inference, in which
bodily states are regulated by autonomic reflexes enslaved by descending
predictions from deep generative models.  Our discomfort signal $D(t)$
serves an analogous role: it is the reservoir's ``interoceptive prediction
error''---the discrepancy between the body's current dynamical state and
its habituated baseline.  \citet{pezzulo2025} recently modeled embodied
decisions as active inference, emphasizing that decision-making in
dynamic environments requires tight coupling between perception, action,
and bodily states; the BRG sentinel provides a concrete computational
mechanism for such coupling.
Individual differences
in interoceptive sensitivity could be modeled by a receptivity precision
parameter modulating $D(t)$; we leave this extension to future work.

\subsection{Relationship to the Free Energy Principle}
\label{sec:discussion_fep}

The free energy framework in this paper shares mathematical structure
with Friston's free energy principle (FEP) \citep{friston2006, friston2010}
and with the thermodynamic decision theory of \citet{ortega2013},
but differs in important respects that require explicit clarification.

\paragraph{Shared structure.}
All three frameworks optimize a functional of the form
``accuracy minus complexity,'' where complexity is measured by a KL
divergence from a reference distribution.  In active inference
\citep{friston2015, parr2022}, the agent minimizes
$\FE = -\E[\ln p(o|\theta)] + \KL(q(\theta) \| p(\theta))$
over a variational posterior $q(\theta)$.
In \citet{ortega2013}, a bounded-rational agent maximizes
$\E[U] - \beta^{-1} \KL(\pi \| \pi_0)$,
where $\pi_0$ is a default policy and $\beta^{-1}$ is the
information-processing cost.
Our $\FE(\alpha) = -\bar{u}(\alpha) + \lambda \cdot \KL(p_{\alpha} \| p_{\mathrm{hab}})$
follows the same pattern.

\paragraph{Key difference: generative model vs.\ dynamical substrate.}
In the FEP, the baseline $p(\theta)$ is a \emph{prior belief} within
a generative model that the agent uses for inference.
In the BRG framework, $p_{\mathrm{hab}}$ is not a belief but the
\emph{physically realized stationary distribution} of the reservoir's
$d$-dimensional state space (the habituated dynamical regime).
The agent does not ``infer'' its habituated distribution;
it \emph{is} its habituated distribution.
This distinction has consequences: FEP agents update beliefs via
variational message passing; BRG agents change their action distribution
by modulating $\alpha$, which shifts the balance between two
\emph{physical processes} (reservoir dynamics and cognitive computation).
The complexity cost in BRG is thus more directly thermodynamic---it
measures the metabolic cost of driving the body away from its adapted
steady state, in the sense of \citet{wolpert2024}'s non-equilibrium
thermodynamics of computation.

\paragraph{Complementarity with multi-agent active inference.}
Recent work has applied active inference to multi-agent game-theoretic
settings.  \citet{hyland2024} proposed free-energy equilibria for
boundedly-rational agents, and \citet{ruizserra2025} extended this with
factorised generative models in which agents maintain explicit,
individual-level beliefs about opponents' internal states.
\citet{demekas2024} provided an analytical model showing that
active-inference agents develop TfT-like strategies in iterated
prisoner's dilemmas through belief updating.
The BRG framework offers a complementary perspective:
rather than modeling opponent types through explicit Bayesian inference,
the body reservoir \emph{implicitly} encodes opponent history through its
$d$-dimensional dynamical state.  This implicit encoding is less flexible
(the reservoir cannot represent arbitrary generative models) but more
efficient (it requires no inference algorithm) and more robust
(it exploits the reservoir's intrinsic smoothing and temporal memory).
A promising direction for future work is a hybrid architecture that
combines reservoir-based implicit encoding with active-inference-based
explicit belief updating, potentially inheriting the robustness of the
former and the flexibility of the latter.

\paragraph{Relationship to entropy-regularized RL.}
Our free energy functional also connects to entropy-regularized
reinforcement learning \citep{haarnoja2018}, where policies are optimized
to maximize reward while staying close to a reference (maximum-entropy)
policy.  The key difference is that our reference distribution
$p_{\mathrm{hab}}$ is \emph{adapted} (shaped by cooperative experience),
not maximum-entropy (uniform); the BRG agent minimizes deviation from
\emph{what it has learned to be}, not from ignorance.

\subsection{Implications for Statistical Physics}

From a statistical physics perspective, the BRG agent is a non-equilibrium
system driven by its game-theoretic environment.  The reservoir's stationary
state distribution $p_{\alpha,\mathcal{E}}$ is a non-equilibrium steady state
maintained by the continuous flow of actions and observations.  The
complexity cost $\KL(p_{\alpha,\mathcal{E}} \| p_{\mathrm{hab}})$ is the
thermodynamic cost of maintaining this steady state relative to the
habituated equilibrium, connecting to \citet{still2012}'s framework for
the thermodynamics of prediction.

Recent advances in stochastic thermodynamics provide a finer-grained
foundation for this connection.  \citet{wolpert2024} argues that real
computers---biological and digital---operate far from thermal equilibrium
under constraints that Landauer's idealized framework does not capture:
finite speed, limited degrees of freedom, and irreversible intermediate
steps.  \citet{manzano2024} formalize ``mismatch cost'' as the excess
dissipation incurred when a computation's actual thermodynamic trajectory
deviates from the optimally designed protocol.  In our framework, the
state-space KL divergence $\KL(p_{\alpha} \| p_{\mathrm{hab}})$ can be
interpreted as a mismatch cost: the thermodynamic penalty for operating
at receptivity $\alpha$ in a noisy environment when the body has been
designed (habituated) for a cooperative one.  This cost is
measured in the reservoir's $d$-dimensional state space, not in the
one-dimensional action space.  An agent that overrides its body
(reducing $\alpha$) distorts its internal state trajectory even when
the resulting behavioral change is small---the body pays for the
override internally.
The free-energy-optimal $\alpha^{*}$ is then the
governance mode that minimizes this internal mismatch cost subject to the
payoff constraint.

The self-consistency equation \eqref{eq:self_consistency} is formally
analogous to mean-field equations in spin glasses and neural networks
\citep{ganguli2008}: the ``magnetization'' (body output $a^{*}$) is
determined self-consistently by the ``effective field'' (reservoir dynamics
plus self-feedback).  The cooperative and defection fixed points correspond
to ferromagnetic and anti-ferromagnetic phases, with the habituation
history playing the role of the external field that selects between them.

The dimension sweep results add a new physical insight: the body's
smoothing power scales with its number of degrees of freedom, analogous
to the law of large numbers in statistical mechanics.
A larger reservoir averages over more independent dynamical modes, producing
a more stable macroscopic output.  This connects to recent findings in
physical reservoir computing, where \citet{lee2024rc} demonstrated that
reservoirs with access to different thermodynamic phases can reconfigure
their computational properties on demand---a physical analog of the
BRG agent's ability to shift between governance modes via $\alpha(t)$.

\subsection{Virtue and Self-Interest: A Brief Remark}
\label{sec:virtue}

The BRG framework suggests an alignment between habitual cooperation and
thermodynamic efficiency that merits brief mention, with the caveat that
the following observations are derived from a single-opponent, fixed-schedule
model and should not be over-generalized.

Within the scope of our analysis, the body-governed cooperator
does not sacrifice payoff for moral principle; rather, cooperation
\emph{is} the thermodynamically efficient response of an adapted body.
The dynamic sentinel ensures that this behavior
is not na\"ively exploitable: when the body detects threat, it activates cognitive
resources for self-protection.
The embodied cooperator's behavior is difficult to distinguish
from genuine commitment---the body does not signal the contingent nature
of its cooperation---giving it properties of a credible commitment device
\citep{frank1988}.  Whether this alignment between virtue and efficiency
holds in richer settings (multi-agent, adaptive opponents, asymmetric
information) remains an open question that deserves separate treatment.

\subsection{Limitations and Future Work}
\label{sec:limitations}

Several limitations suggest directions for future research:

\begin{enumerate}[label=(\roman*)]
  \item \textbf{Limited strategy comparison:}
    Our primary cognitive baseline is continuous Tit-for-Tat, which is
    maximally noise-sensitive.  Experiment~10
    (Section~\ref{sec:exp_ema_baseline}) partially addresses this by
    comparing against EMA-filtered TfT, showing that the reservoir's
    advantage lies not in variance reduction alone (EMA achieves up to
    $158\times$) but in the combination of smoothing, perturbation
    absorption, instant recovery, and anomaly detection---capabilities
    that a one-dimensional filter cannot provide.
    A broader comparison across strategy taxonomies, including Generous
    TfT, Win-Stay-Lose-Shift, quantal response equilibria
    \citep{mckelvey1995}, and level-$k$ models, would further clarify the
    reservoir's added value and is deferred to future work.

  \item \textbf{Single opponent model:}  Our analysis considers a single
    opponent (cooperative, noisy, or defecting).  Extending to populations
    of heterogeneous agents, evolutionary dynamics
    \citep{nowak2006, kandori1993}, and multi-agent reservoir interactions
    is an important direction.

  \item \textbf{Strategic opponents:}  All opponents in our experiments
    follow fixed schedules rather than adaptive strategies.  Testing the
    BRG agent against learning opponents (e.g., reinforcement learners,
    active-inference agents \citep{demekas2024},
    other BRG agents) would assess robustness in more realistic settings.
    In particular, two BRG agents with different $\alpha$ values or
    habituation histories could exhibit emergent coordination or
    competition dynamics that the current single-agent analysis cannot
    capture.

  \item \textbf{Reservoir architecture:}  We use a standard ESN with $\tanh$
    nonlinearity.  More biologically realistic reservoir models (spiking
    networks, dendritic computation) may exhibit different smoothing
    properties.  The theoretical framework applies to any contractive
    recurrent map, but the quantitative predictions are architecture-dependent.

  \item \textbf{Symmetric training inputs:}  During developmental learning,
    the reservoir is driven with symmetric states ($a = a_{\mathrm{opp}}$)
    representing mutual cooperation or mutual defection.  Asymmetric
    input combinations (e.g., cooperate against defector) are not included
    in the training set; the agent encounters them only during game play.
    Whether alternative training protocols improve or degrade robustness
    remains to be investigated.

  \item \textbf{Analytical tightness:}  The smoothing bound in
    Theorem~\ref{thm:smoothing} uses the operator norm $\|J_\Phi\|$
    (which upper-bounds $\rho(J_\Phi)$) and a first-order linearization
    that is conservative for large perturbations.  Tighter bounds using
    structured random matrix theory or the actual spectral radius
    could sharpen the variance reduction estimate.

  \item \textbf{Empirical validation:}  The model makes testable predictions
    about human cooperation: habitual cooperators should show lower
    physiological variability when facing noisy opponents, and their
    cooperation should be less affected by isolated defections.

  \item \textbf{Sentinel learning:}  The sentinel parameters
    ($\alpha_0, \theta, \eta_{\uparrow}, \eta_{\downarrow}$) are fixed.
    A natural extension is meta-learning of these governance parameters
    based on long-term performance, connecting to multi-timescale learning
    in reinforcement learning.

  \item \textbf{Baseline drift:}  The sentinel's discomfort signal uses
    baselines ($\bar{\xstate}, \bar{a}_{\mathrm{body}}$) that are either
    fixed at initialization or tracked by EMA.  In environments with
    gradual regime shifts, EMA-tracked baselines will drift, potentially
    masking slow-onset threats.  The current model is best suited to
    environments with stationary cooperative baselines punctuated by
    discrete perturbations.

  \item \textbf{Receptivity precision:}  We model the discomfort signal
    $D(t)$ as received without attenuation.  A natural extension is a
    \emph{receptivity precision} parameter $\rho_{\mathrm{rec}} \in [0,1]$
    that modulates the signal: $D_{\mathrm{eff}} = \rho_{\mathrm{rec}} \cdot D$.
    This would formalize individual differences in interoceptive sensitivity,
    but its identification from behavior alone is difficult (it is partially
    confounded with the threshold $\theta$), so we defer it to future work.
\end{enumerate}

\subsection{Broader Significance}

The BRG framework offers a resolution to a longstanding puzzle in social
science: why cooperative behavior is so robust in practice despite the
theoretical fragility of cooperative equilibria.  The answer, in our
framework, is that cooperation does not need to be \emph{maintained} by
ongoing cognitive computation; it can be \emph{embodied} in the adapted
state of a dynamical system.  The body-governed cooperator is not deciding
to cooperate each round---cooperation is the natural output of a system
that has been shaped by its history of cooperative interaction.

The dynamic sentinel ensures that this embodied cooperation is not na\"ive.
The body's own discomfort signals provide an early warning system that
activates cognitive resources when needed, without requiring a sophisticated
metacognitive detector.  The result is an agent that is simultaneously
robust (through reservoir smoothing), responsive (through the sentinel),
and efficient (through body governance).

This perspective connects to \citet{ostrom1990}'s observation that
institutions for collective action succeed when they become habits rather
than rules, to \citet{batson2011}'s argument that genuine altruism exists
alongside strategic cooperation, and to the game-theoretic commitment
literature \citep{frank1988, nesse2001}: the most credible commitment is
one that is not computed but expressed---not a strategy but a state of
being.

\section{Conclusion}
\label{sec:conclusion}

We have introduced the Body-Reservoir Governance (BRG) framework for
analyzing cooperation in repeated games from an embodied, thermodynamic
perspective.  The three-layer architecture reinterprets the roles of body,
cognition, and metacognition: the body reservoir is an implicit inferrer
whose $d$-dimensional dynamics encode interaction history and express
cooperation as a self-consistent fixed point; the cognitive filter is
an available toolkit activated on demand; and metacognition is a lightweight
governor that maintains policy without engaging in detection or strategy.

Our central finding is that the body reservoir performs implicit
high-dimensional inference over interaction history, and that this
inference reduces the complexity cost of cooperation as measured by the
state-space KL divergence from the habituated baseline---a cost borne
internally by the reservoir.
The externally observable consequence is dramatic smoothing of behavioral
variance, but the underlying mechanism is richer: the reservoir's
$d$-dimensional state constitutes an implicit model of the interaction
that far exceeds what the one-dimensional action output reveals.
The dynamic sentinel model formalizes the body's dual role: the same
reservoir dynamics that perform implicit inference also detect
environmental change through a composite discomfort signal, driving
adaptive modulation of metacognitive receptivity.  The sentinel achieves
the highest cumulative payoff across all conditions tested, demonstrating
that body-driven governance outperforms both static body trust and pure
cognitive control.

The reservoir dimension sweep establishes that implicit inference capacity
scales with bodily richness: variance reduction ranges from $23\times$ at
$d = 5$ to $1600\times$ at $d = 75$, providing a formal account of why
organisms with richer embodiment can sustain more robust habitual behavior.
Overriding the body---reducing $\alpha$ against the reservoir's
self-consistent tendency---incurs a state-space distortion cost that
increases with the degree of override, formalizing the thermodynamic
expense of acting against one's adapted nature.

The framework suggests a reinterpretation of cooperation.  What appears
as unconditional cooperation from the game-theoretic perspective is, from
the thermodynamic perspective, the most \emph{efficient} form of strategic
behavior: the minimum-dissipation response of an adapted body to its
environment.  The body-governed cooperator expresses the self-consistent
fixed point of a dynamical system shaped by cooperative experience,
with the full richness of that expression encoded in $d$ dimensions
but projected onto one.  When the environment changes, the body's own
discomfort signals trigger the transition to cognitive override.
The framework thus reinterprets cooperation not as a computed strategy
but as an emergent property of embodied dynamics---one that the
sentinel mechanism can modulate when circumstances demand.

\bigskip
\noindent\textbf{Acknowledgments.}
The author thanks the members of the computational social science seminar
at The Open University of Japan for helpful discussions.

\bibliographystyle{plainnat}
\bibliography{references}

\end{document}